\documentclass[review]{elsarticle}

\usepackage{tikz}
\usepackage{xspace}
\usepackage{adjustbox}
\usepackage{amssymb}
\usepackage{amsthm}
\usepackage{amsmath}
\usepackage{complexity}
\usepackage{mathrsfs}
\usepackage{graphicx}
\usepackage{caption}
\usepackage[skip = 0cm, list = true]{subcaption}
\usepackage{algorithm}
\usepackage{algpseudocode}
\usepackage{enumitem}
\usepackage{tabularx}
\usepackage{environ}
\usepackage{multirow}
\usepackage{multicol}
\usepackage{xcolor}
\usepackage{pgfplots}
\usepackage{listings}
\usepackage{lineno,hyperref}
\modulolinenumbers[5]
\usetikzlibrary{
    backgrounds, shapes, arrows, decorations.pathmorphing, decorations.pathreplacing
}
\DeclareMathAlphabet{\mathpzc}{OT1}{pzc}{m}{it}
\newtheorem{theorem}{Theorem}

\newtheorem{lemma}[theorem]{Lemma}
\newtheorem{corollary}[theorem]{Corollary}
\newtheorem{proposition}[theorem]{Proposition}

\newtheorem{problem*}[theorem]{Problem}
\newtheorem{claim}[theorem]{Claim}

\makeatletter

\newcommand{\problemtitle}[1]{\gdef\@problemtitle{#1}}
\newcommand{\probleminstance}[1]{\gdef\@probleminstance{#1}}
\newcommand{\problemquestion}[1]{\gdef\@problemquestion{#1}}
\NewEnviron{problem}{
    \problemtitle{}\probleminstance{}\problemquestion{}
    \BODY
    \par\addvspace{.5\baselineskip}
    \noindent
    \begin{tabularx}{\textwidth}{@{\hspace{\parindent}} l X c}
        \multicolumn{2}{@{\hspace{\parindent}}l}{\underline{\textsc{\@problemtitle}}} \\
        \textbf{Instance:} & \@probleminstance \\
        \textbf{Question:} & \@problemquestion
        \end{tabularx}
    \par\addvspace{.5\baselineskip}
}
\makeatother

\tikzset{
    vertex/.style = {
        draw, 
        circle,
        inner sep = 3pt,
        fill = black,
        outer sep = 1.5pt,
        align = center,
    }, 
    cpi/.append style = {
        fill = white, 
        font = \scriptsize, 
        minimum width = 15,
        inner sep = 0.5pt, 
        thick
    },
    pi/.style = {
        draw = gray,
        regular polygon,
        regular polygon sides = 4,
        fill = white, 
        font = \tiny,
        inner sep = 0.1pt
    }
}

\newcommand{\vgn}{\chi_{_V}^{\text{g}}}
\newcommand{\gapstr}{\text{str}_{\text{gap}}}
\newcommand{\etal}{et~al.\@\xspace}
\newcommand{\vmax}{v_{\text{max}}}
\newcommand{\vmin}{v_{\text{min}}}
\newcommand{\N}{\mathbb{N}}
\newcommand{\I}{\mathpzc{I}}
\newcommand{\X}{\mathpzc{X}}
\newcommand{\Y}{\mathpzc{Y}}
\newcommand{\Z}{\mathpzc{Z}}
\renewcommand{\R}{\mathpzc{R}}
\renewcommand{\O}{\mathcal{O}}

\newcommand{\gkvl}[1]{gap-$[#1]$-vertex-labelling\xspace}
\newcommand{\gvl}{gap-vertex-labelling\xspace}

\newcommand{\dist}{\text{dist}}

\journal{Discrete Applied Mathematics}

\begin{document}

\begin{frontmatter}

\title{Graphs without gap-vertex-labellings:\\families and bounds\tnoteref{mytitlenote}}
\tnotetext[mytitlenote]{Supported by: Grant 2015/11937-9, S\~{a}o Paulo Research Foundation (FAPESP); Grants 425340/2016-3 and 308689/2017-8, National Council for Scientific and Technological Development (CNPq).}

\author{C. A. Weffort-Santos\corref{mycorrespondingauthor}}
\author{R. C. S. Schouery}
\address{Institute of Computing -- University of Campinas (UNICAMP)\\
  Av. Albert Einstein, 1251 --- 13083-852 --- Campinas --- SP --- Brazil\\
  \texttt{\{celso.santos, rafael\}@ic.unicamp.br}}



\cortext[mycorrespondingauthor]{Corresponding author}


\begin{abstract}
A proper labelling of a graph $G$ is a pair $(\pi, c_\pi)$ in which $\pi$ is an assignment of numeric labels to some elements of $G$, and $c_\pi$ is a colouring induced by $\pi$ through some mathematical function over the set of labelled elements. In this work, we consider gap-vertex-labellings, in which the colour of a vertex is determined by a function considering the largest difference between the labels assigned to its neighbours. We present the first upper-bound for the vertex-gap number of arbitrary graphs, which is the least number of labels required to properly label a graph. We investigate families of graphs which do not admit any gap-vertex-labelling, regardless of the number of labels. Furthermore, we introduce a novel parameter associated with this labelling and provide bounds for it for complete graphs $K_n$.
\end{abstract}

\begin{keyword}
Gap-strength\sep gap-labellings\sep proper-labellings
\MSC[2010] 05-XX\sep  05C78
\end{keyword}

\end{frontmatter}


\section{Introduction}\label{sec:intro}

Graph colouring problems have been studied since 1852, when F. Guthrie introduced the Four Colour Problem to his brother and, subsequently, the world~\cite{May67}. Since then, important results in this field have been extended to numerous reoccurring problems such as timetable scheduling, register allocation in compilers, and even solving Sudoku puzzles~\cite{Lewis16}. 

Recently, researchers have turned their attention to a different form of colouring problems, nowadays referred to as graph labellings. In these new versions, numerical values are assigned to some elements of the graph, rather than (simply) colours. Concerning labellings, the value assigned to an element of the graph usually provides some information regarding the modelled problem, e.g.\@ the cost of opening a facility, the frequency assigned to a radio transmitter or the distance between two cities. 

Most authors trace the origins of graph labellings to A. Rosa~\cite{Rosa67}, who defined a $\beta$-valuation of a graph $G$ with $m$ edges as an injection $f : V(G) \rightarrow \{0, 1, \ldots, m\}$ such that $f$ induces another injection $g : E(G) \rightarrow \{1, 2, \ldots, m\}$, in which each edge $e = uv$ is assigned label $g(e) = |f(u) - f(v)|$. Rosa's $\beta$-valuations\footnote{Nowadays, $\beta$-valuations are referred to as ``Graceful Labellings'', thanks to S. Golomb~\cite{Golomb72}.} were the starting point of an entire field of research within Graph Theory. Since then, many different types of labellings have been proposed, each of which make use of different mathematical properties between the labelled elements. As examples, we cite: irregular assignments, harmonious labellings, AVD-labellings, magic and anti-magic labellings. For detailed surveys on this field, we refer the reader to the works of B. D. Acharya et al.~\cite{AcharyaAR08}, A.~Marr~\& W.~Wall~\cite{MarrW13}, P. Zhang~\cite{Zhang15}, J. Gallian~\cite{Gallian18} and S. L\'{o}pez \& F.~Muntaner-Batle~\cite{LopezM17}.

In our examples, a second labelling is often obtained through the use of some mathematical function over the set of labelled elements. Graceful labellings are such an example: Rosa's $\beta$-valuation induces a vertex-distinguishing colouring, in which the colour of each vertex is induced by the largest difference between the labels assigned to its incident edges. These types of label-induced colourings are nowadays known as \emph{proper labellings} and were introduced and studied by M. Karo\'{n}ski et al. in 2004~\cite{KaronskiLT04}.  

In this paper, we investigate gap-$[k]$-vertex-labellings of graphs, a proper labelling which was introduced in 2013 by A. Dehghan et al.~\cite{DehghanSA13}. It is defined as an assignment of integer labels to the vertices of a graph $G$ in such a way that, for every vertex $v \in V(G)$, its colour is induced by the largest difference, i.e. the largest ``gap'', between the labels of its neighbours and, furthermore, the induced colouring is proper. This labelling was inspired by its edge counterpart, which was introduced in 2012 by M. Tahraoui et al.~\cite{TahraouiDK12}, and has been studied both in its computational complexity and its structural properties~\cite{BrandtMNPS16,DehghanSA13,Dehghan16,ScheidweilerT15}. Our work revolves around three new aspects of gap-$[k]$-vertex-labellings. 

First, we consider the problem of deciding whether a graph admits a \gkvl{k}, for some $k \in \N$. We begin by proving a structural property of gap-vertex-labelable graphs, namely that we can restrict our analysis only to gap-vertex-labellings in which all labels (and gaps) are distinct. This property is used to derive the first upper-bound on the vertex-gap number of graphs --- the least $k$ for which a graph admits a \gkvl{k}. Furthermore, we obtain a trivial $\O(n!)$ algorithm which decides whether a graph admits a \gkvl{k}, for $k \in \N$. These results, as well as the basic concepts and notation used throughout the paper, are presented in Section~\ref{sec:preliminaries}.

As a second approach, we consider a statement made by Dehghan et al. in their seminal paper, in which the authors claim that there are graphs which do not admit any \gkvl{k}~\cite{DehghanSA13}, regardless of the number $k$ of labels. Their claim is, in fact, true; however, their paper does not characterize any such families of non-gap-vertex-labelable graphs. The next contribution of this work, presented in Section~\ref{sec:ngkvl}, is to present three families for which the claim is true: complete graphs, powers of paths and powers of cycles.

Finally, in Section~\ref{sec:gapstr}, we consider a new underlying problem on non-gap-vertex-labelable graph. Observing the previously named families, we realized that these graphs are fairly dense, in the sense that the size of the graph is large. We then questioned: how many edges can one remove from a non-gap-vertex-labelable graph such that the resulting graph \underline{still} cannot be labelled? We introduce a novel parameter to better understand this question, which we call the \emph{gap-strength} of a graph, and we investigate it for the family of complete graphs.

\section{Preliminaries}\label{sec:preliminaries}

In this work, all graphs are connected, simple, and finite. For a graph $G$ with vertex set $V(G)$ and edge set $E(G)$, the degree of a vertex $v \in V(G)$ is denoted by $d(v)$ and its (open) neighbourhood, by $N(v)$. The minimum and maximum degree of $G$ are respectively denoted by $\delta(G)$ and $\Delta(G)$.

The distance between two vertices $u$ and $v$ in a graph $G$ is denoted by $\dist_G(u, v)$. A complete graph, a path and a cycle on $n$ vertices are respectively denoted by $K_n$, $P_n$ and $C_n$. The $k$-th \emph{power} of a graph $G$ is the graph $G^k$ with $V(G^k) = V(G)$ and $E(G^k) = \{uv~|~\dist_G(u, v) \leq k\}$. In this paper, we are particularly interested in Powers of Paths and Powers of Cycles, which are respectively denoted $P_n^k$ and $C_n^k$ for such graphs on $n$ vertices. The third and second powers of graphs $P_6$ and $C_5$, respectively, are exemplified in Figure~\ref{fig:powerpathcycle}.\par

\begin{figure}[!htb]
    \centering
    \begin{tikzpicture}
        \begin{scope}[shift = {(0, -1)}]
        \node[vertex] (p0) at (0, 0) {};
        \node[vertex] (p1) at (1, 0) {};
        \node[vertex] (p2) at (2, 0) {};
        \node[vertex] (p3) at (3, 0) {};
        \node[vertex] (p4) at (4, 0) {};
        \node[vertex] (p5) at (5, 0) {};
        \end{scope}

        \draw (p0) to (p1) to (p2) to (p3) to (p4) to (p5);
        \draw (p0) to[bend left = 40] (p2);
        \draw (p1) to[bend left = 40] (p3);
        \draw (p2) to[bend left = 40] (p4);
        \draw (p3) to[bend left = 40] (p5);
        
        \draw (p0) to[bend left = 55] (p3);
        \draw (p1) to[bend left = 55] (p4);
        \draw (p2) to[bend left = 55] (p5);

        \node[vertex] (c0) at (7.75, -0.6) {};
        \node[vertex] (c1) at (9, -0.6) {};
        \node[vertex] (c2) at (10.25, 0) {};
        \node[vertex] (c3) at (9, 0.6) {};
        \node[vertex] (c4) at (7.75, 0.6) {};

        \draw (c0) to (c1) to (c2) to (c3) to (c4) to (c0);
        \draw (c0) to[bend right = 60] (c2);
        \draw (c1) to[out = -30, in = 30, looseness = 5] (c3);
        \draw (c2) to[bend right = 60] (c4);
        \draw (c3) to[out = 130, in = 140, looseness = 2.5] (c0);
        \draw (c4) to[out = 220, in = 230, looseness = 2.5] (c1);
    \end{tikzpicture}
    \caption{Graphs $P_6^3$ and $C_5^2$.}
    \label{fig:powerpathcycle}
\end{figure}
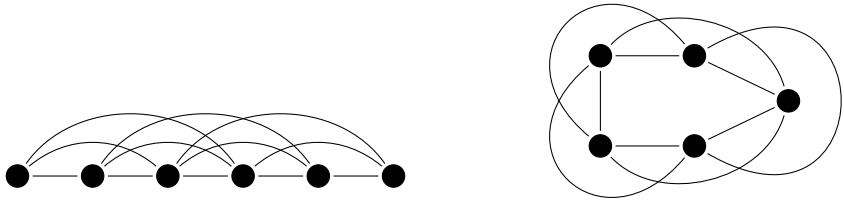

A \emph{proper (vertex-)colouring} of a graph $G$ is an assignment of colours to the vertices of $G$ such that every pair of adjacent vertices receives distinct colours. The least number of colours required to properly colour a graph is called the \emph{chromatic number} of $G$ and is denoted $\chi(G)$. Similarly defined, a \emph{labelling} of~$G$ is an assignment of numeric labels to some elements of $G$. In this work, we consider only label assignments to the vertices of graphs, i.e.\@ \emph{vertex-labellings}.

For a set of labels $[k] = \{1, 2, \ldots, k\}$, a \emph{\gkvl{k}}\footnote{The original definition considers isolated vertices in $G$, assigning colour 1 to them. However, since a graph $G$ is gap-vertex-labelable if and only if all of its connected components are also labelable, we consider only connected graphs, thus removing the case of $d(v) = 0$.} of a graph~$G$ is a pair $(\pi, c_\pi)$, where $\pi : V(G) \rightarrow [k]$ is a vertex-labelling and $c_\pi : V(G) \rightarrow \{0, 1, \ldots, k\}$ is a proper colouring of $G$ such that, for every $v \in V(G)$,
\begin{align}
    c_\pi(v) = \begin{cases}
        \max\limits_{u \in N(v)}\{ \pi(u) \} - \min\limits_{u \in N(v)}\{ \pi(u) \}, &\text{if } d(v) \geq 2; \\
        \pi(u)_{u \in N(v)},&\text{if } d(v) = 1.
    \end{cases}\label{eq:gap}
\end{align}

Note that every vertex $v\in V(G)$ with $d(v) \geq 2$ has its colour induced by the largest difference among the labels assigned to its neighbours; hence the name \emph{gap}-labelling. The least number $k$ of labels for which $G$ \emph{admits} a \gkvl{k} is called the \emph{vertex-gap number} of $G$ and is denoted  $\vgn(G)$. 

As we mention in Section~\ref{sec:intro}, this labelling is a variant of a vertex-distinguish-ing edge-labelling defined by M. Tahraoui~\etal in 2012~\cite{TahraouiDK12}. Gap-$[k]$-vertex-labellings were introduced in 2013 by A. Dehghan~\etal~\cite{DehghanSA13}. In their article, the authors prove that every tree $T$ admits a \gkvl{2}, thus establishing $\vgn(T) = 2$; they also show that every $r$-regular bipartite graph, with $r \geq 4$, is gap-$[2]$-vertex-labelable.

The seminal paper by Dehghan et al. focused on the computational complexity aspects of this labelling. Among other results, they show that the problem of deciding whether a graph $G$ admits a \gkvl{k}, for any fixed $k \geq 2$, is $\NP$-complete. For the particular case of $k = 2$, they show that the problem remains $\NP$-complete even when restricted to the family of 3-colourable graphs and for bipartite graphs. However, they also argue that there is a dichotomy regarding the latter family: if the graph is both bipartite and planar, then the problem lies in $\P$. In 2016, A. Dehghan~\cite{Dehghan16} continued his investigation on the complexity of gap-$[2]$-vertex-labellings of bipartite graphs $G$, proving that the problem of deciding whether $G$ admits a \gkvl{2} such that the induced colouring is a 2-colouring of the graph remains $\NP$-complete.

In their paper, Dehghan~\etal claim that ``a graph may lack any \gkvl{k}''~\cite{DehghanSA13}, and pose the following problem.

\begin{problem*}[Dehghan~\etal, 2013]
    Does there exist a polynomial-time algorithm to determine whether a given graph admits a \gkvl{k}?
\end{problem*}

Motivated by both their claim and the proposed problem, we began our investigation on necessary and sufficient conditions for a graph $G$ to admit a \gkvl{k}. It is important to remark that we are not particularly interested in determining the value of the least $k$ for which a graph admits such a labelling. (Nonetheless, our results enabled us to establish the first upper-bound for the vertex-gap number.) Therefore, from herein, we omit the value~$k$ and refer to graphs that admit a \gkvl{k}, for some $k \in \N$, as \emph{gap-vertex-labelable graphs}; analogously, graphs which do not admit such a labelling are said to be \emph{non-gap-vertex-labelable}. Also, whenever we define a colouring ``as usual'' in the proofs, we mean according to the definition presented in equation~\eqref{eq:gap}. 

\begin{lemma}\label{lemma:different_labels}
    Let $G$ be a graph. Then, $G$ admits admits a \gvl $(\pi, c_\pi)$ if and only if it admits a \gvl $(\pi', c_{\pi'})$ such that, for every pair of distinct vertices $u, v \in V(G)$, $\pi'(u) \neq \pi'(v)$ .
\end{lemma}

\begin{proof}
    Suppose a graph $G$ of order $n$ admits a \gvl $(\pi, c_\pi)$. Note that if every vertex has a distinct label in $\pi$, then the sufficient condition trivially holds. Therefore, in order to prove the result, it suffices to show the necessary condition, i.e., that $G$ admits a \gvl such that every vertex receives a unique label.

    Adjust the notation of $V(G)$ as $\{v_0, \ldots, v_{n-1}\}$ in non-decreasing order, i.e. ${\pi(v_0) \leq \pi(v_1) \leq \ldots \leq \pi(v_{n-1})}$. We define a new labelling $\pi'$ of $G$ as follows. For every vertex ${v_i \in V(G)}$, let $\pi'(v_i) = \pi(v_i)\cdot 2n + i$. Define colouring $c_{\pi'}$ as usual.

    First, we prove that $\pi'$ is a labelling of $G$ such that each vertex received a distinct label. Suppose, for the sake of contradiction, that $\pi'(v_i) = \pi'(v_j)$ for two distinct vertices $v_i, v_j \in V$. Without loss of generality, assume $i < j$. Then,
    \begin{align*}
        \pi'(v_i) = \pi'(v_j) \Rightarrow [\pi(v_i) - \pi(v_j)] \cdot 2n = j - i.
    \end{align*}

    Since $i < j$ by the adjusted notation, the right side of the equation is always positive. However, it is also known that $\pi(v_i) \leq \pi(v_j)$, implying that the left side of the equation is always a negative number. Therefore, there are no values for $i$ and $j$ which satisfy the equation, and we conclude that $\pi'$ is a labelling of $G$ in which every vertex is assigned a distinct label. Furthermore, it is important to remark that $\pi'$ is defined as an \emph{order preserving} function of $\pi$. This means that if $\pi'(v_i) < \pi'(v_j)$ for two vertices $v_i, v_j \in V(G)$, then $\pi(v_i) \leq \pi(v_j)$ in the original \gvl $(\pi, c_\pi)$.

    In order to establish the result, it remains to show that the induced colouring is proper. Suppose there are two adjacent vertices $v_i, v_j \in V$ such that $c_{\pi'}(v_i) = c_{\pi'}(v_j)$. Since the colour of a vertex is induced differently for vertices~$v$ with $d(v) = 1$ and $d(v) \geq 2$, we must address two\footnote{The case $d(v_i) = d(v_j) = 1$ implies that $G \cong K_2$, which can be inspected.} cases: (i) if $d(v_i) \geq 2$ and $d(v_j) \geq 2$; and (ii) if $d(v_i) \geq 2$ and $d(v_j) = 1$. \\

    \paragraph{Case (i). $d(v_i) \geq 2$ and $d(v_j) \geq 2$} Let $v_a$ and $v_b$ be the neighbours of $v_i$ such that $c_{\pi'}(v_i) = \pi'(v_a) - \pi'(v_b)$, and $v_x$ and $v_y$, the neighbours of $v_j$ such that $c_{\pi'}(v_j) = \pi'(v_x) - \pi'(v_y)$; observe that not necessarily vertices $v_a, v_b, v_x$ and $v_y$ are distinct. We express the equality as
    \begin{align}
        c_{\pi'}(v_i) = c_{\pi'}(v_j) &\Rightarrow \pi'(v_a) - \pi'(v_b) = \pi'(v_x) - \pi'(v_y)\nonumber \\
        & \Rightarrow [\pi(v_a) - \pi(v_b) - \pi(v_x) + \pi(v_y)] \cdot 2n = x - y - a + b.\label{eq:lemma:differentlabels:1}
    \end{align}
    From the left side of equation~\eqref{eq:lemma:differentlabels:1}, we consider the two following subcases: if $|\pi(v_a) - \pi(v_b) - \pi(v_x) + \pi(v_y)| \geq 1$; and if $\pi(v_a) - \pi(v_b) - \pi(v_x) + \pi(v_y) = 0$. In the former, since $1 \leq a, b, x, y \leq n$, it follows that $| x - y - a + b| < 2n$, and there are no values for $a, b, x, y$ which satisfy this equation. Therefore, it remains to consider the latter case, in which
    \begin{align}
        \pi(v_a) - \pi(v_b) - \pi(v_x) + \pi(v_y) = 0 \Rightarrow \pi(v_a) - \pi(v_b) = \pi(v_x) - \pi(v_y).\label{eq:lemma:differentlabels:4}
    \end{align}
    Now, recall that $\pi'$ is order preserving, which implies that if $v_a$ and $v_b$ are the vertices that define colour $c_{\pi'}(v_i)$, then $c_\pi(v_i)$ is computed by $\pi(v_a) - \pi(v_b)$; an analogous reasoning holds for $v_j$. Then, we have $\pi(v_a) - \pi(v_b) = c_\pi(v_i)$ and $\pi(v_x) - \pi(v_y) = c_\pi(v_j)$, implying that $c_\pi(v_i) = c_\pi(v_j)$ by equation~\eqref{eq:lemma:differentlabels:4}. This contradicts the fact that $(\pi, c_\pi)$ is a \gvl of $G$, and we conclude that there are no such vertices $v_i$ and $v_j$ with the same induced colour.\\

    {\setlength{\parindent}{0cm}
    \textit{Case (ii). $d(v_i) \geq 2$ and $d(v_j) = 1$.\\}} 
    In this final case, we use a similar reasoning to that of Case (i). First, since $d(v_j) = 1$, its colour is induced by the label assigned to its only neighbour and, hence, $v_x = v_i$. However, note that we may abuse notation and suppose that $v_y$ is a second neighbour of $v_j$ labelled $\pi(v_y) = 0$. Then, the same analysis as that of Case (i) can be applied, and we conclude that $(\pi', c_{\pi'})$ is a \gvl of $G$ in which every vertex receives a distinct label.\qedhere
\end{proof}

With Lemma~\ref{lemma:different_labels} established, we can safely assume that if a graph admits a \gvl, then all labels are distinct. It also allows us to assume that there are exactly two vertices which received the maximum and minimum labels; this particular result will be used extensively in the following sections. Furthermore, we are able to provide the first bounds on the vertex-gap number of arbitrary graphs. 

Let $G$ be a gap-vertex-labelable graph and $(\pi, c_\pi)$, its labelling. As was done in the proof of the previous lemma, consider an ordering ${v_0, v_1, \ldots, v_{n-1}}$ of the vertices of $G$ according to their assigned labels, in increasing order. Observe that by switching the label of every $v_i$ to a distinct powers of two --- namely creating a new labelling $\pi'$ where $\pi'(v_i) = 2^i$, then a similar analysis allows us to conclude that the new colouring induced by $\pi'$ is proper. Thus, $(\pi', c_{\pi'})$ is also a gap-vertex-labelling of $G$, with the added property that the largest label in $\pi'$ is $2^{n-1}$. 

This gives us a first upper-bound on the vertex-gap number of graphs: $\vgn(G) \le 2^{n-1}$. We remark that this somewhat trivial bound is (in some sense) tight since $\vgn(K_3) = 2^{3-1} = 4$; this result is presented in Theorem~\ref{theo:complete}. However, by using the concept of Golomb Rulers, we are able to improve this bound considerably. 

A \emph{Golomb Ruler} of order~$n$ is a set of integers $A = \{a_1, a_2, \ldots, a_n\}$ called \emph{marks}, with $a_1 < a_2 < \ldots < a_n$, such that for each integer $x \ne 0$ there is at most one solution to the equation $x = a_j - a_i$, for marks $a_i, a_j \in A$. Golomb Rulers were introduced independently by W. Babock~\cite{Babock53} and S. Sidon~\cite{Sidon32}; for a detailed survey, we refer the reader to A. Dimitromanolakis' masters thesis.~\cite{Dimitromanolakis02}.

There are simple constructions for Golomb Rulers of order $n$ where the largest mark is at most $\O(n^3)$. However, if $n$ is a prime number, one can construct rulers such that the largest mark is at most $\O(n^2)$~\cite{Ruzsa93,Singer38}. We will use the latter bound since one can quickly find a prime number $p$ between $[n, 2n]$. Furthermore, we explicitly use the Erd\"{o}s-Tur\'{a}n construction, in which it is known that the largest mark is at most $2p^2-p-1$~\cite{ErdosT41}. This allows us to prove the following result.

\begin{theorem}\label{theo:vgnn2}
If $G$ is gap-vertex-labelable, then $\vgn(G) \in \O(n^{2})$.
\end{theorem}

\begin{proof}
    Let $G = (V, E)$ be a simple graph and suppose $G$ admits a \gvl $(\pi, c_\pi)$. By Lemma~\ref{lemma:different_labels}, we adjust notation of $V(G)$ to $\{v_1, v_2, \ldots, v_{n}\}$ such that $\pi(v_1) < \pi(v_2) < \ldots < \pi(v_{n})$.
    Now, define a new labelling $\pi'$ of $V(G)$ as follows. For every $v_i \in V(G)$, assign $\pi'(v_i) = a_{i} + 2p^2$, and define~$c_{\pi'}$ as usual. We prove that $c_{\pi'}$ is a proper colouring of $G$. First, observe that every vertex $v_i$ of degree at least two has its colour induced by some difference $a_l - a_j$ for marks $a_l, a_j$ in the Golomb Ruler. This has two main implications: first, that $c_{\pi'}(v_i) \le 2p^2-p-1$ for every $v_i$ with $d(v_i) \ge 2$. Second, that every pair of adjacent vertices with degree at least two have distinct colours under $c_{\pi'}$, as this is precisely the property that defines the differences between marks in the Golomb Ruler.

    It remains to consider degree-one vertices in the graph. Let $v_i$ be such a vertex and let $v_j$ be its only neighbour in the graph. We safely assume has $v_j$ has degree at least 2, otherwise $G \cong K_2$. Since every vertex in the graph receives a label that is at least $2p^2$, then the induced colour of $v_i$ is also at least~$2p^2$. Given that $v_j$ has degree at least $2$, it follows that $c_{\pi'}(v_j) \le 2p^2-p-1 < 2p^2$. Hence, $v_i$ and $v_j$ have different induced colours. 

    We conclude that there are no adjacent vertices in $G$ with conflicting induced colours and, therefore, that $(\pi', c_{\pi'})$ is a \gvl of $G$. Furthermore, it is a \gvl in which the largest label used is $\O(n^2)$. This completes the proof.
\end{proof}

We conclude this section remarking that one can now design a factorial-time algorithm to decide whether a graph $G$ admits a \gvl. This algorithm consists of assigning every possible combination of marks on the Golomb Ruler to the vertices of~$G$. For each of the $\mathcal{O}(n!)$ assignments, we calculate the induced colours of the vertices and verify if there are any conflicting vertices. Given that determining the induced colour of a vertex and verifying its adjacencies (for conflicting colours) can be done in polynomial time, the following corollary holds.

\begin{corollary}
    There exists a $\O(n!)$-time algorithm which decides whether a given graph~$G$ is gap-vertex-labelable.\hfill$\square$
\end{corollary}

This is the first step towards answering the problem posed by Dehghan et al. in 2013~\cite{DehghanSA13}. However, very little is known about the computational complexity of this decision problem. It would be interesting to pursue a polynomial-time algorithm for this problem or, perhaps, an $\NP$-hardness proof.

\section{\texorpdfstring{Non-gap-vertex-labelable}{} graphs}\label{sec:ngkvl}

In the previous section, we provided necessary and sufficient conditions for a graph to be gap-vertex-labelable and, thus, the first algorithm that decides whether a given graph $G$ admits such a labelling. This reinforces Dehghan et al.'s claim that some graphs ``may lack a gap-vertex-labelling''~\cite{DehghanSA13}.

Nonetheless, we wanted to better understand which structural properties of a graph may be used to determine whether it admits a gap-vertex-labelling or not. In order to do so, we investigated three closely related families of graphs: complete graphs $K_n$, powers of paths $P_n^k$ and powers of cycles $C_n^k$. Our results are presented in the following theorems.

\begin{theorem}\label{theo:complete}
    Let $G \cong K_n$. Then, $G$ is gap-vertex-labelable if and only if $n \leq 3$.
\end{theorem}

\begin{proof}
    Let $G \cong K_n$ be a complete graph of order $n \geq 2$. Complete graph $K_1$ is a trivial graph, for which the result naturally holds. Thus, to prove the sufficient condition, it suffices to show that for the cases $n = 2$ and $n = 3$, $G$ admits some \gvl. These labellings are presented in Figure~\ref{fig:gkvl:k1k2k3}. The numbers inside the vertices correspond to their induced colours, while values in the small grey boxes in their lower-right corners correspond to their assigned labels. (From herein, all figures follow this same representation of labels and induced colours.)

    \begin{figure}[h]
        \centering
        \begin{tikzpicture}
            \node[vertex, cpi] (v0) at (0, 1) {1};
            \node[vertex, cpi] (v1) at (0, -1) {2};
            
            \draw (v0) to (v1);

            \node[vertex, cpi] (w0) at (3, 1) {3};
            \node[vertex, cpi] (w1) at (2, -1) {2};
            \node[vertex, cpi] (w2) at (4, -1) {1};

            \draw (w0) to (w1) to (w2) to (w0);

            \begin{scope}[pi/.append style = {shift = {(0.22, -0.22)}}]
                \node[pi] at (v0) {2};
                \node[pi] at (v1) {1};
                \node[pi] at (w0) {2};
                \node[pi] at (w1) {1};
                \node[pi] at (w2) {4};
            \end{scope}
        \end{tikzpicture}
        \caption{Gap-vertex-labellings of complete graphs $K_2$ and $K_3$.}
        \label{fig:gkvl:k1k2k3}
    \end{figure}
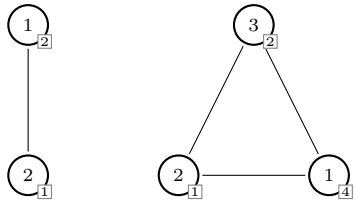

    Now, consider $n \geq 4$, and let $V(G) = \{v_0, \ldots, v_{n-1}\}$ be the vertices of $G$. Suppose $G$ admits a \gvl $(\pi, c_\pi)$. By Lemma~\ref{lemma:different_labels}, we assume all labels are distinct. Adjust notation so that $v_0$ is the vertex which is assigned the largest of all labels in $V(G)$ and $v_1$, the smallest. Consider vertices $v_2$ and $v_3$; these vertices exist in $G$ since $n \geq 4$. Observe that $v_0, v_1 \in N(v_2)$ and $v_0, v_1 \in N(v_3)$. This implies that, regardless of the labels assigned to $v_2, v_3, \ldots, v_{n-1}$, (at least) both $v_2$ and $v_3$ have their colours induced by the same gap, i.e. $c_\pi(v_2) = c_\pi(v_3) = \pi(v_0) - \pi(v_1)$. This is a contradiction since $c_\pi(v_2) \neq c_\pi(v_3)$ in any proper vertex-colouring of $G$. Therefore, $c_\pi$ is not a proper colouring and the result follows. \qedhere
\end{proof}


\begin{theorem}\label{theo:pop}
    Let $G = (V, E)$ be the $k$-th power of path $P_n$, with $n \geq 3$ and $2 \leq k < n$. Then, $G$ is gap-vertex-labelable if and only if: 
    \begin{enumerate}[label = (\roman*)]
        \item $G$ is isomorphic to $P_3^2$ or $P_4^2$; or \label{theo:pop:i}
        \item $n \geq 5$ and $k < n/2$.\label{theo:pop:ii}
    \end{enumerate}
\end{theorem}

\begin{proof}
    Let $G$ be as stated in the hypothesis, with ${V(G) = \{v_0, v_1, \ldots, v_{n-1}\}}$. We begin by exhibiting a \gvl{} $(\pi, c_\pi)$ of $G$ for the sufficient cases. The first graph in Item~\ref{theo:pop:i} is $P_3^2 \cong K_3$, for which Theorem~\ref{theo:complete} establishes the result. For the second, note that assigning labels $(2, 1, 4, 2)$ to vertices $(v_0, v_1, v_2, v_3)$ induces colours $(3, 2, 1, 3)$, respectively, and a quick inspection of this labelling shows that there are no conflicting colours. 

    Concerning Item~\ref{theo:pop:ii}, let $n \geq 5$ and $k < n/2$. For every $0 \leq i \leq n - 1$, assign $\pi(v_i) = 2^i$. Define colouring $c_\pi$ as usual. Figure~\ref{fig:gvl:p83} illustrates this labelling for graphs $P_8^3$ and $P_9^4$. From herein, we will consider the vertices ordered from left to right in increasing index value so as to simplify some statements in the proof.
    
    \begin{figure}[!htb]
        \centering
        \begin{tikzpicture}
            \node[vertex, cpi] (v0) at ({0*1.25}, 0) {6};
            \node[vertex, cpi] (v1) at ({1*1.25}, 0) {15};
            \node[vertex, cpi] (v2) at ({2*1.25}, 0) {31};
            \node[vertex, cpi] (v3) at ({3*1.25}, 0) {63};
            \node[vertex, cpi] (v4) at ({4*1.25}, 0) {126};
            \node[vertex, cpi] (v5) at ({5*1.25}, 0) {124};
            \node[vertex, cpi] (v6) at ({6*1.25}, 0) {120};
            \node[vertex, cpi] (v7) at ({7*1.25}, 0) {48};

            \draw (v0) to (v1) to (v2) to (v3) to (v4) to (v5) to (v6) to (v7);
            \draw (v0) to[bend left = 40, looseness = 1.1] (v2);
            \draw (v1) to[bend left = 40, looseness = 1.1] (v3);
            \draw (v2) to[bend left = 40, looseness = 1.1] (v4);
            \draw (v3) to[bend left = 40, looseness = 1.1] (v5);
            \draw (v4) to[bend left = 40, looseness = 1.1] (v6);
            \draw (v5) to[bend left = 40, looseness = 1.1] (v7);

            \draw (v0) to[bend left = 50, looseness = 1.1] (v3);
            \draw (v1) to[bend left = 50, looseness = 1.1] (v4);
            \draw (v2) to[bend left = 50, looseness = 1.1] (v5);
            \draw (v3) to[bend left = 50, looseness = 1.1] (v6);
            \draw (v4) to[bend left = 50, looseness = 1.1] (v7);

            \begin{scope}[pi/.append style = {shift = {(0.3, -0.35)}}]
                \node[pi] at (v0) {$2^0$};
                \node[pi] at (v1) {$2^1$};
                \node[pi] at (v2) {$2^2$};
                \node[pi] at (v3) {$2^3$};
                \node[pi] at (v4) {$2^4$};
                \node[pi] at (v5) {$2^5$};
                \node[pi] at (v6) {$2^6$};
                \node[pi] at (v7) {$2^7$};
            \end{scope}

            \begin{scope}[shift = {(-0.625, -3)}]
                \node[vertex, cpi] (v0) at ({0*1.25}, 0) {14};
                \node[vertex, cpi] (v1) at ({1*1.25}, 0) {31};
                \node[vertex, cpi] (v2) at ({2*1.25}, 0) {63};
                \node[vertex, cpi] (v3) at ({3*1.25}, 0) {127};
                \node[vertex, cpi] (v4) at ({4*1.25}, 0) {255};
                \node[vertex, cpi] (v5) at ({5*1.25}, 0) {254};
                \node[vertex, cpi] (v6) at ({6*1.25}, 0) {252};
                \node[vertex, cpi] (v7) at ({7*1.25}, 0) {248};
                \node[vertex, cpi] (v8) at ({8*1.25}, 0) {112};

                \draw (v0) to (v1) to (v2) to (v3) to (v4) to (v5) to (v6) to (v7) to (v8);
                \draw (v0) to[bend left = 40, looseness = 1.1] (v2);
                \draw (v1) to[bend left = 40, looseness = 1.1] (v3);
                \draw (v2) to[bend left = 40, looseness = 1.1] (v4);
                \draw (v3) to[bend left = 40, looseness = 1.1] (v5);
                \draw (v4) to[bend left = 40, looseness = 1.1] (v6);
                \draw (v5) to[bend left = 40, looseness = 1.1] (v7);
                \draw (v6) to[bend left = 40, looseness = 1.1] (v8);

                \draw (v0) to[bend left = 50, looseness = 1.1] (v3);
                \draw (v1) to[bend left = 50, looseness = 1.1] (v4);
                \draw (v2) to[bend left = 50, looseness = 1.1] (v5);
                \draw (v3) to[bend left = 50, looseness = 1.1] (v6);
                \draw (v4) to[bend left = 50, looseness = 1.1] (v7);
                \draw (v5) to[bend left = 50, looseness = 1.1] (v8);

                \draw (v0) to[bend left = 60, looseness = 1.1] (v4);
                \draw (v1) to[bend left = 60, looseness = 1.1] (v5);
                \draw (v2) to[bend left = 60, looseness = 1.1] (v6);
                \draw (v3) to[bend left = 60, looseness = 1.1] (v7);
                \draw (v4) to[bend left = 60, looseness = 1.1] (v8);

                \begin{scope}[pi/.append style = {shift = {(0.3, -0.35)}}]
                    \node[pi] at (v0) {$2^0$};
                    \node[pi] at (v1) {$2^1$};
                    \node[pi] at (v2) {$2^2$};
                    \node[pi] at (v3) {$2^3$};
                    \node[pi] at (v4) {$2^4$};
                    \node[pi] at (v5) {$2^5$};
                    \node[pi] at (v6) {$2^6$};
                    \node[pi] at (v7) {$2^7$};
                    \node[pi] at (v8) {$2^8$};
                \end{scope}
            \end{scope}

        \end{tikzpicture}
        \caption{The \gvl{}s of graphs $P_8^3$ and $P_9^4$, as described in the text.}
        \label{fig:gvl:p83}
    \end{figure}
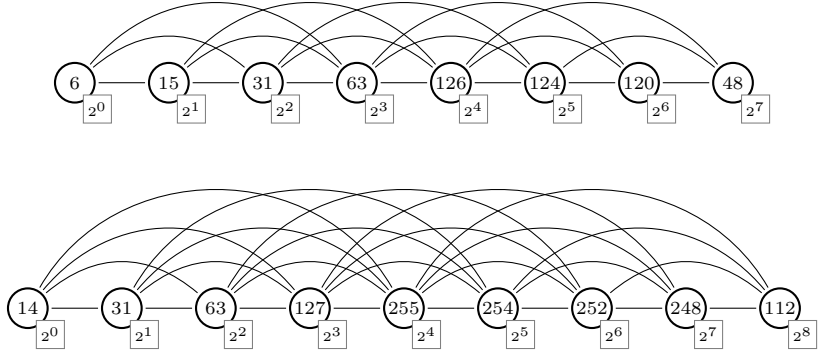 

    In order to prove that $(\pi, c_\pi)$ is a gap-vertex-labelling, it suffices to show that~$c_\pi$ is a proper colouring of $G$. First, we address vertices $v_i$ in the ``middle'' of~$G$, namely with ${i \in [k+1, n - k - 2]}$. For every such $v_i$, its leftmost neighbour is $v_{i-k}$, with ${i - k \geq 1}$, and its rightmost neighbour is $v_{i+k}$, with ${i + k \leq n-2}$. Since the labels are in increasing order (from left to right), it follows that every induced colour ${c_\pi(v_i) = 2^{i+k} - 2^{i-k} = 2^i (2^k - 2^{-k})}$ is distinct. 

    Next, consider vertices $v_j$ with $1 \leq j \leq k$, noting that these are all adjacent to $v_0$. Then, every such $v_j$ has its colours defined by $\pi(v_{j+k}) - \pi(v_0)$. These colours are different from that of their neighbours, whose indices are in ${[k+1, 2k]}$. A similar reasoning applies to the rightmost end of $G$: vertices $v_j$, for ${n - k - 1 \leq j \leq n-2}$, are adjacent to $v_{n-1}$ and, therefore, have $c_\pi(v_j) = 2^{n-1} - 2^{j-k}$. These colours are all distinct amongst themselves and are also different from all that of previously considered vertices. Finally, we draw the reader's attention to the corner cases, namely $c_\pi(v_0) = \pi(v_k) - \pi(v_1) = 2^k - 2$ and $c_\pi(v_{n-1}) = \pi(v_{n-2}) - \pi(v_{n-k-1}) = 2^{n-2} - 2^{n-k-1}$. Since all vertices have their colours induced by the difference between distinct powers of two, every vertex receives a unique colour and, thus, $c_\pi$ is a proper colouring of $G$.

    Now, it remains to show the necessary condition. For $n = 4$, the only graph is $G \cong P_4^3$, for which Theorem~\ref{theo:complete} states there is no \gvl. Thus, we consider $n \geq 5$ and $k \geq n/2$; consequently, $k \ge 3$. We prove this result by contradiction, supposing that there exists a \gvl $(\pi, c_\pi)$ of~$G$ in this case. By Lemma~\ref{lemma:different_labels}, there are two unique vertices which have the maximum and minimum labels in $\pi$; we denote these vertices by $\vmax$ and $\vmin$, respectively. Also, let $i, j$ be their respective indices in the adjusted notation and, without loss of generality, we assume $j > i$. 

    We begin by observing that vertices $\vmax$ and $\vmin$ cannot be ``too close'' to each other, otherwise the result would immediately hold. Formally, we claim that $j - i \ge 4$. To see that this claim is true, we draw the readers attention to Figure~\ref{fig:nge8:case1}. Subfigures (a), (b) and (c) respectively consider the cases where $j - i$, i.e. the number of edges between $\vmin$ and $\vmax$ (strictly) in the path, are equal to 1, 2 and 3. 

    \begin{figure}[!h]
        \centering
        \begin{subfigure}[b]{0.45\textwidth}
            \centering
            \begin{tikzpicture}
                \node (v0)  at (0, 0) {$\ldots$};
                \node[vertex, label = below:{$v_{i - 1}$}] (v1)  at (1, 0) {};
                \node[vertex, label = below:{$\vmax$}] (v2)  at (2, 0) {};
                \node[vertex, label = below:{$\vmin$}] (v3)  at (3, 0) {};
                \node[vertex, label = below:{$v_{j + 1}$}] (v4)  at (4, 0) {};
                \node (v5)  at (5, 0) {$\ldots$};

                \draw (v0) to (v1) to (v2) to (v3) to (v4) to (v5);
                \draw (v0) to[bend left = 45] (v2);
                \draw (v1) to[bend left = 45] (v3);
                \draw (v2) to[bend left = 45] (v4);
                \draw (v3) to[bend left = 45] (v5);

                \draw[orange, thick] (v1) to[bend left = 60, looseness = 1.2] (v4);
            \end{tikzpicture}
            \caption{}\label{fig:nge8:case1:a}
        \end{subfigure}
            \hspace{0.05\textwidth}
        \begin{subfigure}[b]{0.45\textwidth}
            \centering
            \begin{tikzpicture}
                \node (v0)  at (0, 0) {$\ldots$};
                \node[vertex, label = below:{$\vmax$}] (v1)  at (1, 0) {};
                \node[vertex, label = below:{$v_{i+1}$}] (v2)  at (2, 0) {};
                \node[vertex, label = below:{$\vmin$}] (v3)  at (3, 0) {};
                \node[vertex, label = below:{$v_{j+1}$}] (v4)  at (4, 0) {};
                \node (v5)  at (5, 0) {$\ldots$};

                \draw (v0) to (v1) to (v2) to (v3) to (v4) to (v5);
                \draw (v0) to[bend left = 45] (v2);
                \draw (v1) to[bend left = 45] (v3);
                \draw[orange, thick]  (v2) to[bend left = 45] (v4);
                \draw (v3) to[bend left = 45] (v5);

                \draw (v1) to[bend left = 60, looseness = 1.2] (v4);
            \end{tikzpicture}
            \caption{}\label{fig:nge8:case1:b}
        \end{subfigure}

            \vspace{1em}

        \begin{subfigure}[b]{0.45\textwidth}
            \centering
            \begin{tikzpicture}
                \node (v0)  at (0, 0) {$\ldots$};
                \node[vertex, label = below:{$\vmax$}] (v1)  at (1, 0) {};
                \node[vertex, label = below:{$v_{i + 1}$}] (v2)  at (2, 0) {};
                \node[vertex, label = below:{$v_{j - 1}$}] (v3)  at (3, 0) {};
                \node[vertex, label = below:{$\vmin$}] (v4)  at (4, 0) {};
                \node (v5)  at (5, 0) {$\ldots$};

                \draw (v0) to (v1) to (v2); \draw[orange, thick] (v2) to (v3); \draw (v3) to (v4) to (v5);
                \draw (v0) to[bend left = 45] (v2);
                \draw (v1) to[bend left = 45] (v3);
                \draw (v2) to[bend left = 45] (v4);
                \draw (v3) to[bend left = 45] (v5);

                \draw (v1) to[bend left = 60, looseness = 1.2] (v4);
            \end{tikzpicture}
            \caption{}\label{fig:nge8:case1:c}
        \end{subfigure}
        \caption{Cases where the indices of vertices $\vmax$ and $\vmin$ differ in at most 3. The ends of the highlighted edges have conflicting colours.}
        \label{fig:nge8:case1}
    \end{figure}
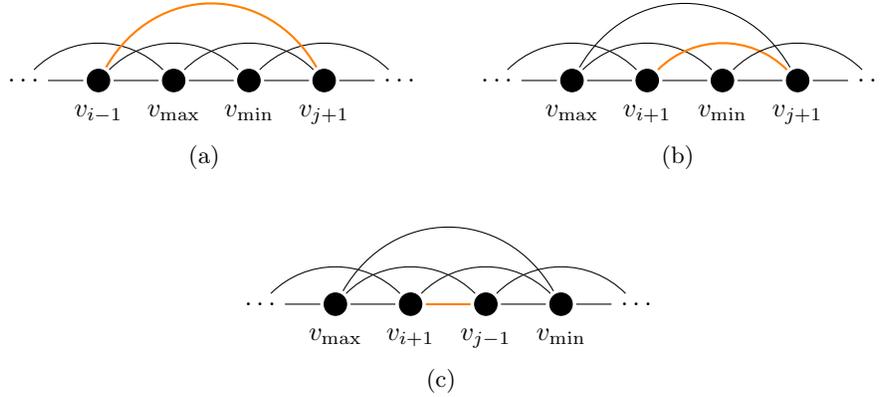

    Now, since $k \ge n/2$, it follows that $n - 1 < 2k$. Furthermore, $j - i \leq n - 1$. This, however, implies that
    \begin{align}
        j - i \leq n - 1 < 2k &\implies i - j +2k > 0\nonumber\\
            &\implies (i + k) - (j - k) + 1 > 1.\label{thisfuckingeq}
    \end{align}

    Observe that $i + k$ is the index of the rightmost vertex adjacent to $\vmax$ (according to our left-to-right orientation) and, similarly, $v_{j - k}$ is the leftmost vertex adjacent to $\vmin$, as represented in Figure~\ref{fig:thisfuckingfig}. Moreover, the left side of Equation~\ref{thisfuckingeq} calculates the number of vertices ``between'' $\vmax$ and $\vmin$ in $P_n$. Thus, the inequality shows that there are at least two vertices between $\vmax$ and $\vmin$ in $G$; clearly these vertices are adjacent to both $\vmax$ and $\vmin$ and, therefore, have the same induced colour. This contradicts the fact that $(\pi, c_\pi)$ is a \gvl of $G$, which completes the proof.\qedhere

    \begin{figure}
        \centering
        \begin{tikzpicture}
            \node (v-1) at (1, 0) {$\ldots$};
            \node[vertex, cpi, label = below:{\scriptsize $\vmax$}] (v0) at (2, 0) {};
            \node (v1) at (3, 0) {$\ldots$};
            \node[vertex, cpi, label = below:{\scriptsize $v_{j-k}$}] (v2) at (4, 0) {};
            \node[vertex, cpi, label = below:{\scriptsize $v_{i+k}$}] (v6) at (6, 0) {};
            \node (v7) at (7, 0) {$\ldots$};
            \node[vertex, cpi, label = below:{\scriptsize $\vmin$}] (v8) at (8, 0) {};
            \node (v9) at (9, 0) {$\ldots$};

            \draw (v-1) to (v0) to (v1) to (v2) to node[pos = 0.5, fill = white] {$\ldots$} (v6) to (v7) to (v8) to (v9);
            \draw (v0) to[bend left = 40, looseness = 1.1] (v2);
            \draw (v2) to[bend left = 40, looseness = 1.1] (v6);
            \draw (v6) to[bend left = 40, looseness = 1.1] (v8);

            \draw (v0) to[bend left = 60, looseness = 1.1] (v6);
            \draw (v2) to[bend left = 60, looseness = 1.1] (v8);
            
            \draw[decoration = {brace, mirror}, decorate, thick] (3.5, -0.75) to (6.5, -0.75);
        \end{tikzpicture}
        \caption{A representation of vertices $\vmax$, $\vmin$ and their respective neighbours. The contradiction is highlighted in the centre of the image. Most edges and vertices were omitted.}
        \label{fig:thisfuckingfig}
    \end{figure}
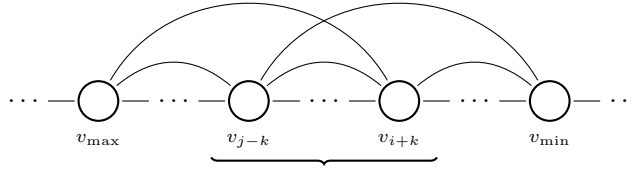
\end{proof} 

Through a similar analysis, we are able to show necessary and sufficient conditions for the family of Powers of Cycles, which are presented in the following theorem.

\begin{theorem}\label{theo:poc}
    Let $G = (V, E)$ be the $k$-th power of cycle $C_n$, with $n \geq 4$ and $2 \leq k < n/2$. Then, $G$ is gap-vertex-labelable if and only if: 
    \begin{enumerate}[label = (\roman*)]
        \item $G$ is isomorphic to $C_6^2$ or $C_7^2$; or \label{theo:poc:i}
        \item $n \geq 8$ and $k \le \lfloor n/4\rfloor$.\label{theo:poc:ii}
    \end{enumerate}
\end{theorem}

\begin{proof}
    Let $G$ be as stated in the hypothesis, with $V(G) = \{v_0, \ldots, v_{n-1}\}$. We begin by considering graphs with $4 \le n \le 7$ vertices. Note that since $C_4^2 \cong K_4$, $C_5^2 \cong K_5$, $C_6^3 \cong K_6$ and $C_7^3 \cong K_7$, Theorem~\ref{theo:complete} establishes that these graphs are non-gap-vertex-labelable. Thus, for the considered values of $n$, it remains to show that $C_6^2$ and $C_7^2$ admit gap-vertex-labellings; these are shown in Figure~\ref{fig:gvl:c62c72}. 

    \begin{figure}[!htb]
        \begin{subfigure}[b]{0.45\textwidth}
            \centering
            \begin{tikzpicture}
                \def \radius {2};
                \node[vertex, cpi] (v0) at (60:\radius) {2};
                \node[vertex, cpi] (v1) at (120:\radius) {3};
                \node[vertex, cpi] (v2) at (180:\radius) {1};
                \node[vertex, cpi] (v3) at (240:\radius) {2};
                \node[vertex, cpi] (v4) at (300:\radius) {3};
                \node[vertex, cpi] (v5) at (0:\radius) {1};

                \draw (v0) to (v1) to (v2) to (v3) to (v4) to (v5) to (v0);
                \draw (v0) to[bend left = 15] (v2);
                \draw (v1) to[bend left = 15] (v3);
                \draw (v2) to[bend left = 15] (v4);
                \draw (v3) to[bend left = 15] (v5);
                \draw (v4) to[bend left = 15] (v0);
                \draw (v5) to[bend left = 15] (v1);

                \begin{scope}[pi/.append style = {shift = {(0.22, -0.22)}}]
                    \node[pi] at (v0) {1};
                    \node[pi] at (v1) {2};
                    \node[pi] at (v2) {4};
                    \node[pi] at (v3) {1};
                    \node[pi] at (v4) {2};
                    \node[pi] at (v5) {4};
                \end{scope}
            \end{tikzpicture} 
            \caption{}
            \label{fig:gvl:c62}
        \end{subfigure}
            \hspace{0.05\textwidth}
        \begin{subfigure}[b]{0.45\textwidth}
            \centering
            \begin{tikzpicture}
                \def \radius {2};
                \begin{scope}[rotate = 90]
                    \node[vertex, cpi] (v0) at ({(360/7)*0}:\radius) {6};
                    \node[vertex, cpi] (v1) at ({(360/7)*1}:\radius) {3};
                    \node[vertex, cpi] (v2) at ({(360/7)*2}:\radius) {7};
                    \node[vertex, cpi] (v3) at ({(360/7)*3}:\radius) {4};
                    \node[vertex, cpi] (v4) at ({(360/7)*4}:\radius) {2};
                    \node[vertex, cpi] (v5) at ({(360/7)*5}:\radius) {3};
                    \node[vertex, cpi] (v6) at ({(360/7)*6}:\radius) {7};
                \end{scope}

                \draw (v0) to (v1) to (v2) to (v3) to (v4) to (v5) to (v6) to (v0);
                \draw (v0) to[bend left = 15] (v2);
                \draw (v1) to[bend left = 15] (v3);
                \draw (v2) to[bend left = 15] (v4);
                \draw (v3) to[bend left = 15] (v5);
                \draw (v4) to[bend left = 15] (v6);
                \draw (v5) to[bend left = 15] (v0);
                \draw (v6) to[bend left = 15] (v1);

                \begin{scope}[pi/.append style = {shift = {(0.22, -0.22)}}]
                    \node[pi] at (v0) {$1$};
                    \node[pi] at (v1) {$8$};
                    \node[pi] at (v2) {$4$};
                    \node[pi] at (v3) {$4$};
                    \node[pi] at (v4) {$4$};
                    \node[pi] at (v5) {$4$};
                    \node[pi] at (v6) {$2$};
                \end{scope}
            \end{tikzpicture} 
            \caption{}
            \label{fig:gvl:c72}
        \end{subfigure}
        \caption{Gap-vertex-labellings of graphs $C_6^2$ and $C_7^2$ in (a) and (b), respectively.}
        \label{fig:gvl:c62c72}
    \end{figure}
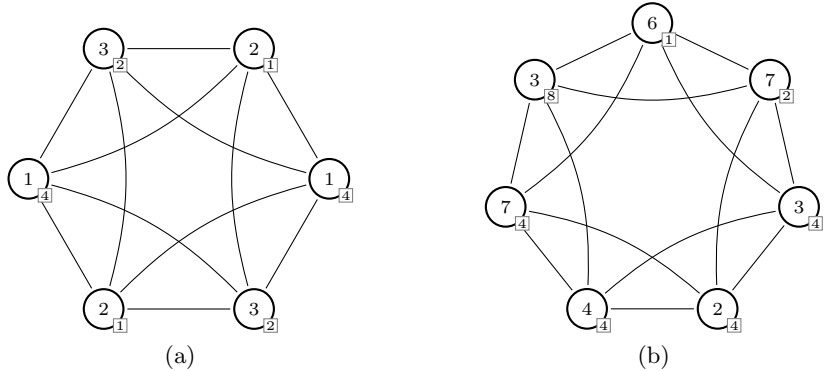

    Next, consider $G \cong C_n^k$ with $n \ge 8$ and $2 \le k \le \lfloor n/4\rfloor$. We show $G$ is gap-vertex-labelable by defining a labelling $\pi$ of $V(G)$ as follows: for every vertex $v_i$, $0 \le i < n/2$, let $\pi(v_i) = 2^i$ (we will refer to these vertices as vertices in the \emph{right side} of the cycle); and for the remaining vertices $v_j$, $j \ge n/2$ (in the \emph{left side} of the cycle), let $\pi(v_j) = 2^{\lceil n/2\rceil + (n - j)}$. Define colouring $c_\pi$ as usual. This labelling and its induced colouring are exemplified for graphs $C_{8}^2$ and $C_{9}^2$ in Figure~\ref{fig:gvl:c82c92}. 

    \begin{figure}[!htb]
        \centering
        \begin{subfigure}[b]{0.45\textwidth}
            \begin{tikzpicture}[rotate = 90, yscale = -1]
                \def \radius {2.5};
                \node[vertex, cpi, label = {\scriptsize $\vmin$}] (v0) at ({45*0}:\radius) {$60$};
                \node[vertex, cpi] (v1) at ({45*1}:\radius) {$31$};
                \node[vertex, cpi] (v2) at ({45*2}:\radius) {$255$};
                \node[vertex, cpi] (v3) at ({45*3}:\radius) {$254$};
                \node[vertex, cpi, label = below:{\scriptsize $\vmax$}] (v4) at ({45*4}:\radius) {$124$};
                \node[vertex, cpi] (v5) at ({45*5}:\radius) {$248$};
                \node[vertex, cpi] (v6) at ({45*6}:\radius) {$255$};
                \node[vertex, cpi] (v7) at ({45*7}:\radius) {$127$};

                \draw (v0) to (v1) to (v2) to (v3) to (v4) to (v5) to (v6) to (v7) to (v0);
                \draw (v0) to[bend left = 30] (v2);
                \draw (v2) to[bend left = 30] (v4);
                \draw (v4) to[bend left = 30] (v6);
                \draw (v6) to[bend left = 30] (v0);
                \draw (v1) to[bend left = 30] (v3);
                \draw (v3) to[bend left = 30] (v5);
                \draw (v5) to[bend left = 30] (v7);
                \draw (v7) to[bend left = 30] (v1);
            
                \begin{scope}[pi/.append style = {shift = {(0.3, -0.3)}}]
                    \node[pi] at (v0) {$2^0$};
                    \node[pi] at (v1) {$2^1$};
                    \node[pi] at (v2) {$2^2$};
                    \node[pi] at (v3) {$2^3$};
                    \node[pi] at (v4) {$2^8$};
                    \node[pi] at (v5) {$2^7$};
                    \node[pi] at (v6) {$2^6$};
                    \node[pi] at (v7) {$2^5$};
                \end{scope}
            \end{tikzpicture}
            \caption{}\label{fig:gvl:c82}
        \end{subfigure}
            \hspace{0.05\textwidth}
        \begin{subfigure}[b]{0.45\textwidth}
            \begin{tikzpicture}[rotate = 90, yscale = -1]
                \def \radius {2.5};
                \node[vertex, cpi, label = {\scriptsize $\vmin$}] (v0) at ({40*0}:\radius) {124};
                \node[vertex, cpi] (v1) at ({40*1}:\radius) {63};
                \node[vertex, cpi] (v2) at ({40*2}:\radius) {15};
                \node[vertex, cpi] (v3) at ({40*3}:\radius) {510};
                \node[vertex, cpi] (v4) at ({40*4}:\radius) {508};
                \node[vertex, cpi, label = below:{\scriptsize $\vmax$}] (v5) at ({40*5}:\radius) {248};
                \node[vertex, cpi] (v6) at ({40*6}:\radius) {496};
                \node[vertex, cpi] (v7) at ({40*7}:\radius) {511};
                \node[vertex, cpi] (v8) at ({40*8}:\radius) {255};

                \draw (v0) to (v1) to (v2) to (v3) to (v4) to (v5) to (v6) to (v7) to (v8) to (v0);
                \draw (v0) to[bend left = 30] (v2);
                \draw (v2) to[bend left = 30] (v4);
                \draw (v4) to[bend left = 30] (v6);
                \draw (v6) to[bend left = 30] (v8);
                \draw (v8) to[bend left = 30] (v1);
                \draw (v1) to[bend left = 30] (v3);
                \draw (v3) to[bend left = 30] (v5);
                \draw (v5) to[bend left = 30] (v7);
                \draw (v7) to[bend left = 30] (v0);

                \begin{scope}[pi/.append style = {shift = {(0.3, -0.3)}}]
                    \node[pi] at (v0) {$2^0$};
                    \node[pi] at (v1) {$2^1$};
                    \node[pi] at (v2) {$2^2$};
                    \node[pi] at (v3) {$2^3$};
                    \node[pi] at (v4) {$2^4$};
                    \node[pi] at (v5) {$2^9$};
                    \node[pi] at (v6) {$2^8$};
                    \node[pi] at (v7) {$2^7$};
                    \node[pi] at (v8) {$2^6$};
                \end{scope}
            \end{tikzpicture}
            \caption{}\label{fig:gvl:c92}
        \end{subfigure}
        \caption{Examples of gap-vertex-labellings, as described in the text, for graphs $C_8^2$ and $C_9^2$ in (a) and (b), respectively. We name the vertices which received the largest and smallest labels in $\pi$ as $\vmax$ and $\vmin$, respectively.}
        \label{fig:gvl:c82c92}
    \end{figure}
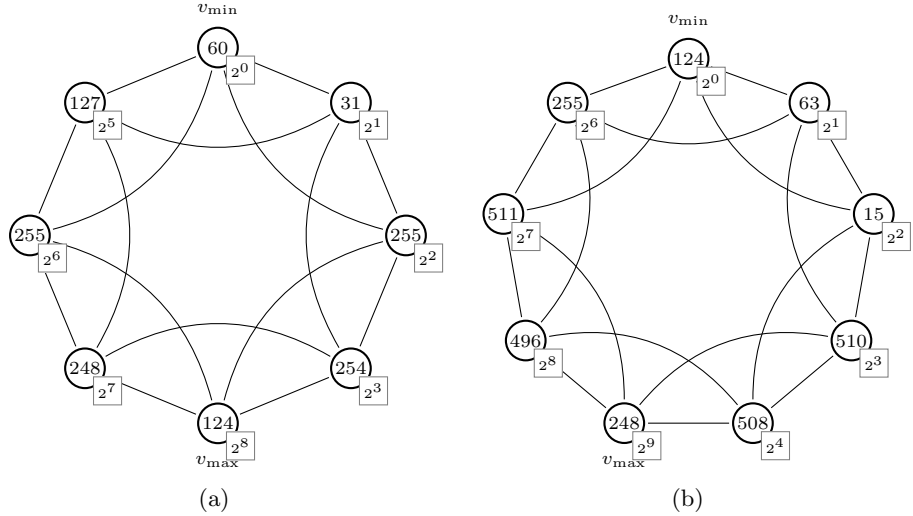

    We prove that $(\pi, c_\pi)$ is a gap-vertex-labelling of $G$ by showing that $c_\pi$ is a proper colouring of $G$. From herein, all operations on the \underline{indices} of the vertices are taken modulo $n$. Note that this has no effect when considering the values of the \underline{labels} assigned to these vertices. Also, we remark that the smallest and highest value labels are assigned to vertices $v_0$ and $v_{\lceil n/2\rceil}$, respectively, and we will refer to these vertices as $\vmin$ and $\vmax$. We consider three different cases, depending on the indices of the vertices.

    \paragraph{Case (i)} First, consider vertices $v_1, v_2, \ldots, v_k$, observing that these vertices are all adjacent to $\vmin$. For every such vertex $v_i$, its colour is induced by $v_0$ and their furthest neighbour in the left side of the cycle, which receives label $2^{n-k + i}$. The same reasoning can be applied to vertices $v_{n-k}, v_{n-k+1}, \ldots, v_{n-1}$, which are also adjacent to $\vmin$: their colours are induced by their furthest neighbour in the left side (also) and by $\vmin$. Since these induced colours are all induced by the difference between distinct powers of two, no two vertices in these ranges receive the same induced colour.

    \paragraph{Case (ii)} Next, consider vertices $v_{\lceil n/2\rceil - 1}, \ldots, v_{\lceil n/2\rceil - k}$, which are located in the right side of the cycle and are all adjacent to $\vmax$. Their colours are induced by $\vmax$ and by their furthest neighbour in the right side, which receive decreasingly smaller labels. Thus, for every such vertex $v_i$, its induced colour is $2^{n} - 2^{i - k}$. Once again, a similar reasoning applies to the vertices in the left side of the cycle that are adjacent to $\vmax$, namely vertices $v_{\lceil n/2\rceil + 1}, \ldots, v_{\lceil n/2\rceil + k}$. Given that all colours are induced by distinct differences of powers of two, their induced colours are all different.

    \paragraph{Case (iii)} Finally, we consider the remaining vertices in the right and left sides of $G$. These are vertices $v_{k+1}, \ldots, v_{\lceil n/2\rceil - (k+1)}$ and $v_{\lceil n/2\rceil + k + 1}, \ldots$, $v_{n - (k+1)}$, respectively. Since these vertices are not adjacent to either $\vmax$ or $\vmin$, it follows that the colour induced in these vertices already differs from those considered in cases (i) and (ii). Then, for every such vertex $v_i$, its colour is induced by $2^{i+k} - 2^{i-k}$ and, therefore, they are all distinct amongst themselves. We conclude that no two vertices receive the same colour under $\pi$ and, therefore, $(\pi, c_\pi)$ is a gap-vertex-labelling of $G$. 

    Thus, in order to complete the proof of the theorem, it remains only to show the necessary condition. We accomplish this by considering $n \ge 8$ and, in order to achieve a contradiction, suppose $k > \lfloor n/4\rfloor$. Observe that ${3 \le k \le \lfloor n/2\rfloor}$ since $n \ge 8$. Now, suppose $(\pi, c_\pi)$ is a \gvl of $G$ and, by Lemma~\ref{lemma:different_labels}, let $\vmax$ and $\vmin$ be the vertices of $G$ which received the largest and smallest labels in $\pi$, respectively. Also, as was done in the proof of Theorem~\ref{theo:pop}, let $i$ and $j$ be the respective indices of $\vmax$ and $\vmin$; without loss of generality, let $i < j$. Once again, vertices $\vmax$ and $\vmin$ cannot be too close to each other, and we safely assume $j - i \ge 4$.

    Considering the symmetries of the graph, it follows that $j - i \le \lfloor n/2\rfloor \le  n/2$. Since $k > \lfloor n/4\rfloor$ and $k$ is integer, we can safely assume $k \ge \lfloor n/4\rfloor + 1 \ge  n/4$. Thus, $2k \ge  n/2$ and, similar to the proof of Theorem~\ref{theo:pop}, we have
    \begin{align}
        j - i \le n/2 < 2k &\implies i - j + 2k > 0\nonumber\\
            &\implies (i + k) - (j - k) + 1 > 1.\label{eq:poc:i}
    \end{align}

    Once again, $(i + k)$ and $(j - k)$ are the respective indices of the furthest vertices in $G$ that are adjacent to $\vmax$ and $\vmin$, and that are located ``in between'' these two vertices in $V(G)$. Furthermore, equation~\eqref{eq:poc:i} states that there are at least two vertices in between $v_{i + k}$ and $v_{j - k}$ which are adjacent to each other and to both $\vmax$ and $\vmin$, implying they have the same induced colour. Thus, $c_\pi$ is not a proper colouring of $G$ --- a contradiction.\qedhere
\end{proof}

\section{The \texorpdfstring{gap-strength}{} of complete graphs}\label{sec:gapstr}
In this section, we present a novel approach to non-gap-vertex-labelable graphs, recalling that such is the case for complete graphs $K_n$ of order $n \ge 4$; we will consider solely these graphs for the rest of the paper. 

As illustrated in the proofs in Section~\ref{sec:ngkvl}, a graph lacks any gap-vertex-labelling if there are two vertices whose designated labels would define the induced colour of two adjacent vertices in their joint neighbourhood. For complete graph $K_4$, this situation is illustrated in the left side of Figure~\ref{fig:k4k4m1}: vertices $\vmax$ and $\vmin$ would define the colours of both $u$ and $v$, which leads to a conflict. Notice, however, that removing edge $uv$ in this graph would resolve this conflict, and thus the resulting graph admits a \gkvl{4}. 

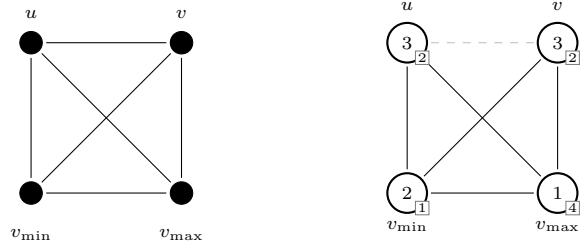
\begin{figure}[htb]
    \centering
    \begin{tikzpicture}
        \node[vertex, label = {[shift = {(0, -0.95)}]\scriptsize{$\vmin$}}] (v0) at (0, 0) {};
        \node[vertex, label = {\scriptsize{$u$}}] (v1) at (0, 2) {};
        \node[vertex, label = {\scriptsize{$v$}}] (v2) at (2, 2) {};
        \node[vertex, label = {[shift = {(0, -0.95)}]\scriptsize{$\vmax$}}] (v3) at (2, 0) {};
        \draw (v0) to (v1) to (v2) to[thick] (v3) to (v0);
        \draw (v0) to (v2); \draw (v1) to (v3);

        \node[vertex, cpi, label = {[shift = {(0, -0.95)}]\scriptsize{$\vmin$}}] (v4) at (5, 0) {$2$};
        \node[vertex, cpi, label = {\scriptsize{$u$}}] (v5) at (5, 2) {$3$};
        \node[vertex, cpi, label = {\scriptsize{$v$}}] (v6) at (7, 2) {$3$};
        \node[vertex, cpi, label = {[shift = {(0, -0.95)}]\scriptsize{$\vmax$}}] (v7) at (7, 0) {$1$};

        \draw (v4) to (v5); \draw (v6) to (v7) to (v4);
        \draw (v4) to (v6); \draw (v5) to (v7);
        \draw (v5) edge[gray!50, dashed] (v6);

        \begin{scope}[pi/.append style = {shift = {(0.2, -0.2)}}]
            \node[pi] at (v4) {\tiny{$1$}};
            \node[pi] at (v7) {\tiny{$4$}};
            \node[pi] at (v5) {\tiny{$2$}};
            \node[pi] at (v6) {\tiny{$2$}};
        \end{scope}
    \end{tikzpicture}
    \caption{Graph $K_4$ and the graph obtained by removing edge $uv$ from $K_4$; to the right, a \gkvl{4} of the latter.}
    \label{fig:k4k4m1}
\end{figure}

Hence, by removing one edge from $K_4$, we have a gap-vertex-labelable graph. Could this simple operation be enough for complete graph $K_5$ as well? Consider the graph $G'$ obtained by removing an arbitrary edge from $K_5$, which is depicted in Figure~\ref{fig:k5m1}. Suppose this graph admits a \gvl and let $\vmax$ be an arbitrary vertex which receives the largest label. Now, if $\vmin$ is adjacent to $\vmax$, as illustrated in Figure~\ref{fig:k5m1:a}, then the endpoints of the highlighted edges have both $\vmax$ and $\vmin$ in their respective neighbourhoods. This implies that, regardless of the labels assigned to these vertices, they all have the same induced colour $\pi(\vmax) - \pi(\vmin)$. Therefore, $\vmax$ and $\vmin$ are not adjacent. This second case is illustrated in Figure~\ref{fig:k5m1:b} and, once again, the highlighted edges indicate three vertices which have the same induced colour. Therefore, this graph does not admit a \gvl, and we conclude that removing only one edge is not enough for this particular graph.

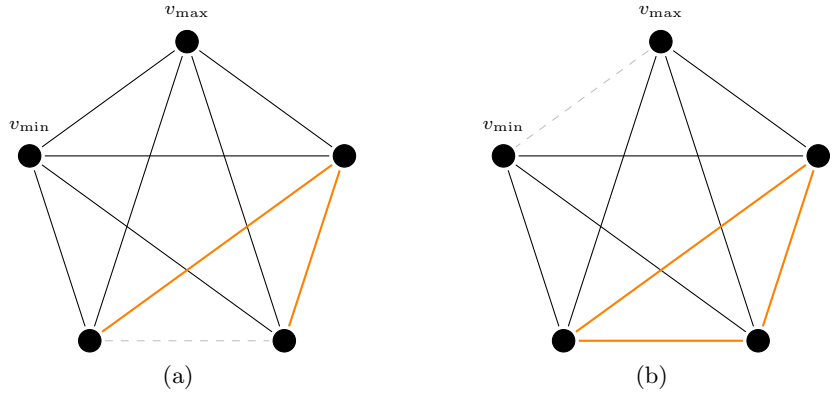
\begin{figure}[htb]
    \centering
    \begin{subfigure}{0.45\textwidth}
        \centering
        \begin{tikzpicture}
            \node[vertex, label = {\scriptsize $\vmax$}] (v0) at (90:2.2) {};
            \node[vertex, label = {\scriptsize $\vmin$}] (v1) at (162:2.2) {};
            \node[vertex] (v2) at (234:2.2) {};
            \node[vertex] (v3) at (306:2.2) {};
            \node[vertex] (v4) at (18:2.2) {};

            \draw (v0) edge (v1) edge (v2) edge (v3) edge (v4);
            \draw (v1) edge (v2) edge (v3) edge (v4);
            \draw (v2) edge[gray!50, dashed] (v3) edge[thick, orange] (v4);
            \draw (v3) edge[thick, orange] (v4);
        \end{tikzpicture}
        \caption{}
        \label{fig:k5m1:a}
    \end{subfigure}
        \hspace{0.05\textwidth}
    \begin{subfigure}{0.45\textwidth}
        \centering
        \begin{tikzpicture}
            \node[vertex, label = {\scriptsize $\vmax$}] (v0) at (90:2.2) {};
            \node[vertex, label = {\scriptsize $\vmin$}] (v1) at (162:2.2) {};
            \node[vertex] (v2) at (234:2.2) {};
            \node[vertex] (v3) at (306:2.2) {};
            \node[vertex] (v4) at (18:2.2) {};

            \draw (v0) edge[gray!50, dashed] (v1) edge (v2) edge (v3) edge (v4);
            \draw (v1) edge (v2) edge (v3) edge (v4);
            \draw (v2) edge[thick, orange] (v3) edge[thick, orange] (v4);
            \draw (v3) edge[thick, orange] (v4);
        \end{tikzpicture}
        \caption{}
        \label{fig:k5m1:b}
    \end{subfigure}
    \caption{Graph $K_5$ without an edge. In (a), $\vmax$ and $\vmin$ are adjacent, while this is not the case in (b).}
    \label{fig:k5m1}
\end{figure}

This raises an immediate follow-up question: what if we remove two edges from $K_5$? Here, we note that two distinct graphs can be obtained by this operation: the first, by removing a maximum matching of $K_5$; and the second, by removing two adjacent edges. These graphs are illustrated in figures~\ref{fig:k5m2:a} and~\ref{fig:k5m2:b}, respectively, and note that both graphs admit \gvl{}s. 

\begin{figure}[htb]
    \centering
    \begin{subfigure}{0.45\textwidth}
        \centering
        \begin{tikzpicture}
            \node[vertex, cpi] (v0) at (90:2.2) {$2$};
            \node[vertex, cpi] (v1) at (162:2.2) {$3$};
            \node[vertex, cpi] (v2) at (234:2.2) {$3$};
            \node[vertex, cpi] (v3) at (306:2.2) {$1$};
            \node[vertex, cpi] (v4) at (18:2.2) {$1$};

            \draw (v0) edge (v1) edge (v2) edge (v3) edge (v4);
            \draw (v1) edge[gray!50, dashed] (v2) edge (v3) edge (v4);
            \draw (v2) edge (v3) edge (v4);
            \draw (v3) edge[gray!50, dashed] (v4);

            \node[pi] at ([shift = {(0.22, -0.22)}]v0) {$1$};
            \node[pi] at ([shift = {(0.22, -0.22)}]v1) {$2$};
            \node[pi] at ([shift = {(0.22, -0.22)}]v2) {$2$};
            \node[pi] at ([shift = {(0.22, -0.22)}]v3) {$4$};
            \node[pi] at ([shift = {(0.22, -0.22)}]v4) {$4$};
        \end{tikzpicture}
        \caption{}
        \label{fig:k5m2:a}
    \end{subfigure}
        \hspace{0.03\textwidth}
    \begin{subfigure}{0.3\textwidth}
        \centering
        \begin{tikzpicture}
            \node[vertex, cpi] (v0) at (90:2.2) {$5$};
            \node[vertex, cpi] (v1) at (162:2.2) {$7$};
            \node[vertex, cpi] (v2) at (234:2.2) {$8$};
            \node[vertex, cpi] (v3) at (306:2.2) {$3$};
            \node[vertex, cpi] (v4) at (18:2.2) {$5$};

            \draw (v0) edge[gray!50, dashed] (v1) edge (v2) edge (v3) edge[gray!50, dashed] (v4);
            \draw (v1) edge (v2) edge (v3) edge (v4);
            \draw (v2) edge (v3) edge (v4);
            \draw (v3) edge (v4);

            \node[pi] at ([shift = {(0.22, -0.22)}]v0) {$1$};
            \node[pi] at ([shift = {(0.22, -0.22)}]v1) {$4$};
            \node[pi] at ([shift = {(0.22, -0.22)}]v2) {$4$};
            \node[pi] at ([shift = {(0.22, -0.22)}]v3) {$9$};
            \node[pi] at ([shift = {(0.22, -0.22)}]v4) {$2$};
        \end{tikzpicture}
        \caption{}
        \label{fig:k5m2:b}
    \end{subfigure}
    \caption{The graphs obtained by: (a) removing a maximal matching of $K_5$; and (b) removing two adjacent edges. Both graphs admit \gvl{}s, which are shown in the figure.}
    \label{fig:k5m2}
\end{figure}
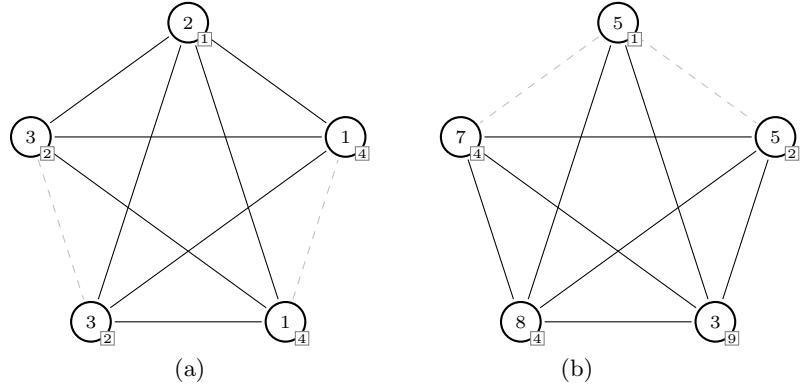

The logical course of action would be to analyse graph $K_6$, and what happens when we remove edges from this graph. By inspecting the graphs obtained by removing one and (any) two edges from $K_6$, we observe that none of these graphs admit a \gvl{}. (Similarly to $K_5$, this conclusion is reached upon inspecting the possible combinations of $\vmax$ and $\vmin$ within the resulting graphs.) However, by removing a perfect matching from $K_6$, we obtain the graph depicted in Figure~\ref{fig:gvl:c62}, which does, in fact, admit a \gkvl{4}. 

Since we are focusing our studies on non-gap-vertex-labelable graphs, the main question to be asked here is: how many edges can one remove from a complete graph $K_n$ such that the resulting graph is \emph{still} non-gap-vertex-labelable. Or, equivalently, what is the least number $l$ of edges that must be removed from~$K_n$ such that there exists some resulting graph which admits a \gvl? We address these questions in this section.

Formally, for a graph $G$, let us define $G^{-l}$ as the family of graphs obtained by removing $l$ edges from $G$, in no particular order. As examples: the rightmost graph in Figure~\ref{fig:k4k4m1} exemplifies the (only) graph in $K_4^{-1}$, while both graphs in Figure~\ref{fig:k5m2} belong to $K_5^{-2}$. Furthermore, we know that no graph in $K_6^{-1}$ or in $K_6^{-2}$ admits a \gvl, whereas there exists a graph in $K_6^{-3}$ which does. It is important to remark that this is not the case for \underline{every} graph in $K_6^{-3}$; for example, the graph obtained by removing three adjacent edges from $K_6$ is non-gap-vertex-labelable. 

With this definition, we are ready to introduce a novel problem associated with non-gap-vertex-labelable graphs and, with it, a novel parameter.

\begin{problem}\label{problem:gapstr}
    \problemtitle{\underline{Gap-strength}}
    \probleminstance{A non-gap-vertex-labelable graph $G$.}
    \problemquestion{What is the least $l$ for which there exists a graph $G' \in G^{-l}$ such that $G'$ is gap-vertex-labelable?}
\end{problem} 

We define the least $l$ that answers the above problem as the \emph{gap-strength} of~$G$, and we denote this parameter by $\gapstr(G)$. To illustrate why we decided to name this parameter using ``strength'' as the keyword, consider graph $K_6$. This graph is sufficiently \emph{strong} that the removal of two edges is not enough to create a gap-vertex-labelable graph. Therefore, graph $K_6$ is relatively ``stronger'' than~$K_4$, for example, since we require the removal of more edges from the former in order to create a gap-vertex-labelable graph. Similarly, by comparing~$K_5$ and~$K_6$, we conclude that $K_5$ is relatively ``weaker''. 

In the following subsections, we provide bounds for the gap-strength of complete graphs $K_n$, $n \geq 4$. We investigate a lower-bound in Section~\ref{subsec:lb} by using a Dynamic Programming algorithm and formalize the result through an algebraic analysis of the recurrence formula, which yields that $\gapstr(K_n) \in \Omega(n^{6/5})$. In Section~\ref{subsec:ub}, we remove edges in such a way to form what we call a restricted decomposition of $K_n$, showing that the removal of $\O(n^{3/2})$ edges is sufficient for the resulting graph to be gap-vertex-labelable. 

\subsection{A lower-bound on the gap-strength of $K_n$}\label{subsec:lb}


In order to establish bounds on the gap-strength of complete graphs $K_n$, ${n \ge 4}$, let us consider a graph $G\in K_n^{-l}$, with ${l = \gapstr(K_n)}$, and let $(\pi, c_\pi)$ be a \gvl of $G$. By Lemma~\ref{lemma:different_labels}, we can safely assume that there are two vertices $\vmax, \vmin \in V(G)$ which have received the largest and smallest labels in $\pi$, respectively. By observing every vertex $v \in V(G) \setminus \{\vmax, \vmin\}$, we can conclude that only one out of four cases occurs: 

\begin{enumerate}
    \item[(i)] $\vmax, \vmin \in N(v)$; or
    \item[(ii)] $\vmax \in N(v)$ and $\vmin \not\in N(v)$; or, conversely
    \item[(iii)] $\vmin \in N(v)$ and $\vmax \not\in N(v)$; or, finally
    \item[(iv)] $\vmax, \vmin \not\in N(v)$.
\end{enumerate}

We will denote the sets of vertices that satisfy cases (i), (ii), (iii) and~(iv) by $\I, \X, \Y$ and $\Z$, respectively. Note that these sets in conjunction with $\{\vmin, \vmax\}$ partition $V(G)$, and thus we will refer to this partition as a \emph{decomposition} of $G$, denoting it by $G(\X, \Y, \Z, \I)$. Also, define $x = |\X|$, $y = |\Y|$, $z = |\Z|$ and $i = |\I|$, observing that the order of the graph can be determined by $n = x + y + z + i + 2$. An illustration of a decomposition of $G$ is presented in Figure~\ref{fig:g-l:layout}.

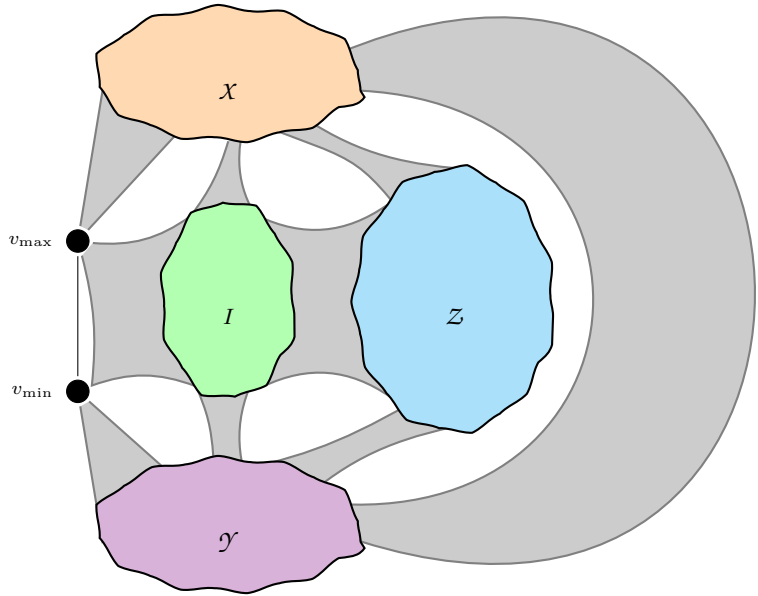
\begin{figure}[htb]
    \centering
    \begin{tikzpicture}
        \clip (-1, -3) rectangle (7, 5.5);
        \node[vertex, label = left:{\scriptsize $\vmax$}] (vmax) at (0, 2) {};
        \node[vertex, label = left:{\scriptsize $\vmin$}] (vmin) at (0, 0) {};

        \begin{scope}[every node/.style = {draw = black, decoration = {snake, segment length = 25, amplitude = 1}, decorate, thick, ellipse, text opacity = 1}]
            \node[fill = orange!30, minimum height = 50, minimum width = 100] (x) at (2, 4) {$\mathpzc{X}$};
            \node[fill = violet!30, minimum height = 50, minimum width = 100] (y) at (2, -2) {$\mathpzc{Y}$};
            \node[fill = cyan!30, minimum height = 100, minimum width = 75] (z) at (5, 1) {$\mathpzc{Z}$};
            \node[fill = green!30, minimum height = 75, minimum width = 50] (i) at (2, 1) {$\mathpzc{I}$};
        \end{scope}

        \draw (vmax) to (vmin);
        \coordinate (xwest) at ([shift = {(0.1, 0)}]x.west);
        \coordinate (xsouth) at ([shift = {(0, 0.22)}]x.south);
        \coordinate (ywest) at ([shift = {(0.1, 0)}]y.west);
        \coordinate (ynorth) at ([shift = {(-0.2, 0.25)}]y.north);
        \coordinate (inorth) at ([shift = {(-0.5, 0)}]i.north);
        \coordinate (isouth) at ([shift = {(-0.43, 0.4)}]i.south);
        \coordinate (znorth) at ([shift = {(-0.9, -0.3)}]z.north);
        \coordinate (zsouth) at ([shift = {(-0.85, 0.7)}]z.south);
        
        \begin{scope}[on background layer]
            \fill[thick, draw = gray, fill = gray!40] (vmax) to (xwest) to ([shift = {(-0.5, 0.3)}]xsouth) to (vmax);
            \fill[thick, draw = gray, fill = gray!40] (vmin) to (ywest) to ([shift = {(-0.5, -0.3)}]ynorth) to (vmin);
            \fill[thick, draw = gray, fill = gray!40] (vmax) to[bend right = 20] (inorth) to (isouth) to[bend right = 20] (vmin) to[bend right = 10] (vmax);
            \fill[thick, draw = gray, fill = gray!40] ([shift = {(0.9, 0)}]inorth) to[bend right = 30] (znorth) to (zsouth) to[bend right = 30] ([shift = {(0.9, 0)}]isouth) to (inorth);
            \fill[thick, draw = gray, fill = gray!40] (xsouth) to[bend left = 10] ([shift = {(-0.5, 0)}]i.north) to ([shift = {(0.5, -0.2)}]i.north) to[bend left = 25] ([shift = {(0.2, 0.2)}]xsouth);
            \fill[thick, draw = gray, fill = gray!40] (x.east) to[bend left = 90, looseness = 2] ([shift = {(-0.3, 0.5)}]y.east) to (y.east) to[bend right = 110, looseness = 3] ([shift = {(-0.3, 0.5)}]x.east);
            \fill[thick, draw = gray, fill = gray!40] (ynorth) to[bend right = 10] (isouth) to ([shift = {(-0.2, 0.3)}]i.south east) to[bend right = 25] ([shift = {(0.5, -0.5)}]ynorth);

            \fill[thick, draw = gray, fill = gray!40] ([shift = {(-1, -0.2)}]x.east) to[bend right = 15] ([shift = {(0, 0.2)}]z.north) to ([shift = {(-0.2, -1)}]z.north) to [bend right = 15, looseness = 1.5] ([shift = {(0.2, 0.3)}]xsouth);
            \fill[thick, draw = gray, fill = gray!40] ([shift = {(-0.2, 0)}]y.north) to[bend right = 10] ([shift = {(-0.6, 0.6)}]z.south) to ([shift = {(0, 0.3)}]z.south) to[bend right = 15] ([shift = {(0.8, -0.5)}]y.north);
        \end{scope}

    \end{tikzpicture}
    \caption{The decomposition $G(\X, \Y, \Z, \I)$ of graph $G$, with each subset of $V$. The grey areas indicate possible edges connecting vertices in distinct sets.}
    \label{fig:g-l:layout}
\end{figure}

We use the idea of decompositions to establish bounds on the gap-strength of $K_n$. Given values for $x, y, z$ and $i$, let $l(x, y, z, i)$ denote the number of edges that must be removed from $K_n$ in order to obtain a gap-vertex-labelable graph with the given decomposition. Observe that one must remove at least: $x$ edges connecting $\vmin$ to each vertex in $\X$; $y$ edges connecting $\vmax$ to vertices in $\Y$, and $2z$ edges connecting vertices in $\Z$ to both $\vmin$ and $\vmax$. Thus, we can write $l(x, y, z, i)$ as follows:
\begin{align}
    l(x, y, z, i) = x + y + 2z + |\R^{\X}| + |\R^{\Y}| + |\R^{\Z}| + |\R^{\I}| + |\R'|.\label{eq:gapstr:l}
\end{align}

Here, each $\R^{S}$ denotes the set of edges removed inside each $S \in \{\X, \Y, \Z, \I\}$, and $\R'$ denotes the set of edges removed between any two distinct sets. Observe that the gap-strength of complete graph $K_n$ can be determined by the following modification to equation~\eqref{eq:gapstr:l}:
\begin{align}
    \gapstr(K_n) = \min_{\substack{x, y, z, i\in \mathbb{Z}_{\ge 0};\\x + y + z+ i + 2 = n}}\{l(x, y, z, i)\}.\label{eq:gapstr:min}
\end{align}   

In the pursuit of a lower-bound for $\gapstr(K_n)$, we note that $|\R'| \ge 0$ and omit this value in equations herein. Now, observe set $\I$, noticing that every vertex $v_i$ in $\I$ has $\vmax, \vmin \in N(v_i)$. This implies that every $v_i$ has the same induced colour and, therefore, set $\I$ must be an independent set in $G$. Therefore, the number of edges removed from within $\I$ is:
\begin{align*}
    |\R^\I| = \binom{i}{2} = \frac{i(i-1)}{2}.
\end{align*}

Next, consider set $\X$. Every $v_x \in \X$ is, by definition, adjacent to $\vmax$ and not adjacent to $\vmin$. If we assume no edges have been removed from within $\X$, then $G' = G[\vmax + \X]$ is a complete graph of order $n' = x + 1$. Now, if $x \ge 3$, Theorem~\ref{theo:complete} states that $G'$ also does not admit a \gvl, and it would require the removal of edges from within $G'$ in order to obtain a gap-vertex-labelable graph. In conclusion, whenever $x \ge 3$ in a decomposition of $G$, we must look at the subproblem of removing edges from a complete (sub)graph of order $n \ge 4$ --- which is the exact same premise upon which a decomposition of $G$ was introduced. 

Note, however, that this decomposition has a somewhat peculiar form. Since every vertex in $G'$ is adjacent to $\vmax$, we must take this property into account when considering a new decomposition of $G'$. This immediately implies that any decomposition $G'(\X', \Y', \Z', \I')$ must have sets $\Y'$ and $\Z'$ empty. Let $l'(x)$ denote the minimum number of edges removed from within set $\X$ in a \emph{restricted} decomposition $G'(\X', \emptyset, \emptyset, \I')$ in order to obtain a gap-vertex-labelable graph. Then:
\begin{align*}
    l'(x) = \min\limits_{x'+i' = x-2}\left\{x' + \binom{i'}{2} + |\R^{\X'}| + |\R'|\right\}.
\end{align*}

Once again, we note that the number of edges removed between sets $|\R'| \ge 0$ and, therefore, it follows that
\begin{align}
    l'(x) \ge \min\limits_{x' + i' = x - 2} \left\{x' + \frac{i(i-1)}{2} + l'(x'+1)\right\}.\label{eq:gapstr:restricted:lb}
\end{align}

We have replaced $|\R^{\X'}|$ from the above equation with $l'(x'+1)$ to account for the recursive decomposition of set $\X'$. In order to establish a lower-bound for $\gapstr(K_n)$, we first demonstrate the following proposition.

\begin{proposition}\label{prop:restricted}
    Let $l'(n)$ be as defined in equation~\eqref{eq:gapstr:restricted:lb}. Then, $l'(n) \in \Omega(n^{3/2})$. More precisely, $l'(n) \ge \frac{1}{10}n^{3/2}$ for all $n \ge 4$.
\end{proposition}

\begin{proof}
    Let $l'(n)$ be as defined in the hypothesis. The objective is to show that, regardless of cardinalities of sets $\X'$ and $\I'$, the main inequality $l'(n) \ge \frac{1}{10}n^{3/2}$ holds for all $n \ge 4$. Note that for complete graph $K_4$, $l'(n) = 1$ by definition and the desired inequality holds. Therefore, we consider $n \ge 5$ from herein.

    Since $i = n  - x - 2$ in a restricted decomposition, let us define an auxiliary function $f(n, x)$ by substituting $i$ from equation~\eqref{eq:gapstr:restricted:lb} as follows:
    \begin{align*}
        f(n, x) &= x + \frac{(n - x - 2)(n-x-3)}{2} + l'(x+1).
    \end{align*}

    \noindent Note that $l'(n) \ge \min\limits_{0 \le x \le n-2}\{f(n, x)\}$ and, therefore, it suffices to show that $f(n, x) \ge \frac{1}{10}n^{3/2}$ for all $n \ge 5$ and $x \in \{0, \ldots, n-2\}$. We do so by, first, considering the case where very few vertices are assigned to set $\X'$, namely when $x \in \{0, 1, 2\}$. Then, 
    \begin{align*}
        f(n, x) &= x + \frac{(n - x - 2)(n-x-3)}{2}
    \end{align*}

    \noindent in this case since $l'(x+1) = 0$ by definition for this range of values of $x$. For ${x \in \{0, 1, 2\}}$, one can prove by induction that $f(n, x) \ge \frac{1}{10}n^{3/2}$ for all $n \ge 4$ and, thus, the desired result holds. Hence, it remains to consider $x \ge 3$, and we prove this second part by (strong) induction on $n$.

    Suppose $l'(k) \ge \frac{1}{10}k^{3/2}$ for every $k \le n$, and let $g(n, x)$ be the following function:
    \begin{align*}
        g(n, x) &= x + \frac{(n - x - 2)(n-x-3)}{2} + \frac{1}{10}(x+1)^{3/2}.
    \end{align*}

    By applying induction hypothesis, it follows that $f(n, x) \ge g(n, x)$. Considering Claim~\ref{claim:lb:restricted:1} from~\ref{app:restricted:claim1}, let $x^{*}$ be the real value that minimizes $g(n, x)$ under $x$. Now, consider a function $h(n)$ which computes the difference between $g(n, x^{*})$ and $\frac{1}{10}n^{3/2}$ for any value of $n$, that is,
    \begin{align*}
        h(n) := g(n, x^{*}) - \frac{1}{10}n^{3/2}.
    \end{align*}

    As verified in Claim~\ref{claim:lb:restricted:2} from~\ref{app:restricted:claim2}, function $h(n)$ is strictly increasing as a function of $n$ and, furthermore, $h(n)$ is nonnegative for all $n \ge 4$. 

    We conclude that $f(n, x) \ge g(n, x) \ge g(n, x^{*}) \ge \frac{1}{10}n^{3/2}$ from Claims~\ref{claim:lb:restricted:1} and~\ref{claim:lb:restricted:2} and, therefore, $l'(n) \in \Omega(n^{3/2})$. This completes the proof.\qedhere
\end{proof}

\bigskip

With Proposition~\ref{prop:restricted} proved, we can now return to analysing equation~\eqref{eq:gapstr:min}. 
\begin{align*}
    \gapstr(K_n) \ge \min_{\substack{x, y, z, i\in \mathbb{Z}_{\ge 0}\\x+y+z+i = n-2}}\left\{x + y + 2z + \frac{i(i-1)}{2} + l'(x+1) + l'(y+1) + \R^\Z\right\}
\end{align*}

Observe that we have lower-bounded $\gapstr(K_n)$ by replacing $\R^\X$ for $l'(x+1)$. We have also replaced $\R^\Y$ with $l'(y+1)$ since a similar argument can be applied to set $\Y$: every $v_y \in \Y$ is, by definition, adjacent to $\vmin$ and not adjacent to~$\vmax$ and, assuming no edges have been removed from within set $\Y$, then the graph induced by $\{\vmin\} \cup \Y$ is a complete graph of order $y+1$ --- which is non-gap-vertex-labelable for $y \ge 3$. With these considerations, we are now ready to prove the main theorem of this section.

\begin{theorem}\label{theo:gapstrength:lowerbound}
    Let $K_n$ be a complete graph of order $n \ge 4$. Then, \[\gapstr(K_n) \in \Omega(n^{6/5}).\]
\end{theorem}

\begin{proof}

    Let $\gapstr(K_n)$ be as defined in equation~\eqref{eq:gapstr:min}. We show that for all ${n \ge 218}$, $\gapstr(K_n) \ge \frac{3}{100}n^{6/5}$. We remark that for the case of $n < 218$, the result was verified by using a Dynamic Programming algorithm which computes the lower-bound considering all possible cardinalities of $\X, \Y, \I$ and $\Z$ in a decomposition. The results obtained by our computer program are presented in~\ref{app:dp}.

    Now, let $q$ be the real value such that $n = q^3 + 2$, noting that $q \ge 6$ since $n \ge 218$. Then, given that $n - 2 = x + y + i + z + 2 = q^3$, our main objective can be restated as showing that
    \begin{align*}
        x + y + 2z + \frac{i(i-1)}{2} + 2z + l'(x+1) + l'(y+1) + \R^\Z \ge \frac{3q^{18/5}}{100}.
    \end{align*}

    As was done in the proof of the lower-bound for the restricted case, we demonstrate the result depending on the size of each part in the decomposition. We begin analysing set $\I$, supposing that $i \ge \left(\frac{3}{10}\right)^{2/3}\cdot q^{12/5}$. In this case, note that for all $q \ge 1.52$, it holds that
    \begin{align}
        \frac{i(i-1)}{2} \ge \frac{3q^{18/5}}{100}.\label{eq:gapstri}
    \end{align}

    An analogous reasoning holds for sets $\X$ and $\Y$. Suppose either of them, say $x \ge \left(\frac{3}{10}\right)^{2/3}\cdot q^{12/5}$. Then we can replace the lower-bound for $l'(x+1)$ given by Proposition~\ref{prop:restricted} since $x \ge 4$ and we obtain
    \begin{align}
        x + \frac{(x+1)^{3/2}}{10} \ge \frac{3q^{18/5}}{100}\label{eq:gapstrx}
    \end{align}

    \noindent for all $q \ge 0$; the same holds for $y$. Thus, we conclude that if either $x, y$ or $i$ is too large, the desired lower-bound immediately holds.

    It remains to consider when $x, y, i < \left(\frac{3}{10}\right)^{2/3}\cdot q^{12/5}$, in which case we have
    \begin{align}
        z \ge q^3-3q^{12/5}\left(\frac{3}{10}\right)^{2/3}.\label{eq:lb:z}
    \end{align}

    \noindent Given that $x, y, i \ge0$ and that the cardinality of set $\Z$ is bounded by equation~\eqref{eq:lb:z}, the proof of our result comes down to showing that
    \begin{align}
        2z + \frac{3z^{6/5}}{100} \ge \frac{3q^{18/5}}{100}.\label{eq:show_z}
    \end{align}

    Here it is important to remark that we have applied an induction hypothesis by replacing $\R^{\Z}$ with $\frac{3}{100}z^{6/5}$. 
    Now, let $w$ be the following function, defined by taking the left side of equation~\eqref{eq:show_z} and dividing it by the right side:
    \begin{align}
        w(q) &= \frac{200z + 3z^{6/5}}{3q^{18/5}}\label{eq:lb:w}
    \end{align}

    By replacing $z$ with the value from equation~\eqref{eq:lb:z} and taking the limit of $w(q)$, we obtain \[\lim_{q\rightarrow \infty} w(q) = 1.\] Therefore, in order to show that inequality~\eqref{eq:show_z} holds, it remains only to show that the function is decreasing with respect to $q$, for $q \ge 6$. Then, take the first derivative of $w$ with respect to $q$:
    \begin{equation}
        \resizebox{\hsize}{!}{
        $w'(q) = \frac{3000\cdot 3^{2/3}\cdot 10^{1/3} - 5000\cdot q^{3/5} + 27\cdot q^{2/3}\cdot 10^{2/15}\left(-3\cdot3^{2/3}\cdot10^{1/3}\cdot q^{12/5}+10q^3\right)^{1/5}}{125q^{11/5}}.$}\nonumber
    \end{equation}

    \noindent Since $125q^{11/5}$ is a positive value, we need only show that the top portion of the fraction is negative. Furthermore, we can ``ignore'' the following term of the equation: ${(-3\cdot 3^{2/3}\cdot 10^{1/3}\cdot q^{12/5})}$; it only helps in decreasing the value of $w'(q)$. Thus, we are left with analysing
    \begin{align}
        3000\cdot 3^{2/3}\cdot 10^{1/3} +\left(- 5000 + 27 \cdot 3^{2/3}\cdot10^{1/3}\right)q^{3/5}.\label{eq:caceta}
    \end{align}

    Let $a = 3000\cdot3^{2/3} 10^{1/3} = 13444.2$ and $b = - 5000 + 27 \cdot 3^{2/3}\cdot10^{1/3} \approx -4879$. Then, we have that equation~\eqref{eq:caceta} can be viewed as $a + bq^{3/5}$, which has negative value for any $q \ge 5.41594$. This implies that $w'(q)$ is decreasing for $q \ge 6$ and, thus, at least 1 for $q \ge 6$. We conclude that inequality~\eqref{eq:show_z} does, in fact, hold for all considered values of $q$. This completes the proof of the lower-bound.
\end{proof}

This result is somewhat astonishing: in order to obtain a gap-vertex-labelable graph $G'$ by removing edges from a complete graph on $n$ vertices, we require the removal of at least $n^{6/5}$ edges. Therefore, Theorem~\ref{theo:gapstrength:lowerbound} shows us that we cannot hope to have a gap-vertex-labelable graph if the number of remaining edges is too large. 

The question that arises is: do we require the removal of much more than $\Omega(n^{6/5})$ edges from $K_n$ such that the resulting graph is gap-vertex-labelable? In the following section, we prove that the removal of $\O(n^{3/2})$ is sufficient.

\subsection{The removal of $\O(n^{3/2})$ edges suffices}\label{subsec:ub}


Let us now turn our discussion towards obtaining an upper-bound on the gap-strength of complete graphs. In Section~\ref{sec:ngkvl}, we considered the family of Powers of Cycles, providing necessary and sufficient conditions for such a graph~$C_n^k$ to admit a gap-vertex-labelling. Observe that the result from Theorem~\ref{theo:poc} already provides us with an upper-bound on $\gapstr(K_n)$ in the sense that a gap-vertex-labelable power of cycle can be obtained by removing $\frac{n(n-2)}{4}$ from complete graph $K_n$ and obtaining a gap-vertex-labelable power of cycle.

We improve on this implicit result by providing an edge-removal algorithm which, given a complete graph $G \cong K_n$, removes $\O(n^{3/2})$ edges and obtains a gap-vertex-labelable graph. Our technique is particularly interesting because the resulting graph is still fairly dense when compared to the power of cycle graphs. This main result is presented in the following theorem. 

\begin{theorem}\label{theo:gapstrength:upperbound}
    Let $G$ be a complete graph of order $n \ge 4$. Then, \[\gapstr(G) \in \O(n^{3/2}).\]
\end{theorem}

\begin{proof}
    Let $K_n$ be a complete graph of order $n \geq 4$. We create a restricted decomposition $G(\X, \emptyset, \emptyset, \I)$ of $K_n$ by a recursive process; we will denote this partitioning simply by $G(\X, \I)$. Each iteration $j$ in our construction will partition the (current) vertex set of a complete graph, which we denote $V_j$, into sets~$\X_j$ and $\I_j$. 

    Let $\vmax$ be an arbitrary vertex in $K_n$, to which we will assign the largest value label --- hence its name. In the first iteration $j = 1$, we have $n_1 = n$ and $V_1 = V(K_n) - \vmax$. For the $j$-th iteration of the construction, partition $V_j$ according to the following rules:

    \begin{itemize}
        \item select an arbitrary vertex from $V_j$ and call it $\vmin^j$; 
        \item select a set $\I_j$ of cardinality $i_j = \lfloor \sqrt{n_j} \rfloor$; and a
        \item set $\X_j$ of cardinality $x_j = n_j - i_j - 2$.
    \end{itemize}

    In this last step, if $x_j \geq 3$, define $n_{j+1} = x_j + 1$, $V_{j+1} = \X_j$ and continue on iteration $j+1$. Otherwise, we end our construction. In Figure~\ref{fig:gapstr:restricted:decomp:k15} we exemplify the first, second and last iterations of our recursive process decomposing complete graph $K_{15}$.

    \begin{figure}[!htb]
        \centering
        \begin{subfigure}[b]{0.45\textwidth}
            \centering
            \begin{tikzpicture}[vertex/.style = {draw, circle, fill = black, outer sep = 1pt, inner sep = 2pt}, scale = 0.8]
                \node[vertex, label = left:{\scriptsize $\vmax$}] (vmax) at (0, 1.5) {};
                \node[vertex, label = below:{\scriptsize $\vmin^1$}] (vmin1) at (1, 0) {};

                \begin{scope}[on background layer]
                    \fill[thick, draw = black, gray!40] (vmax) to (0.3, 4.5) to (5, 2.5) to (vmax);

                    \fill[thick, draw = black, gray!40] (vmax) to (4, 1.15) to (3.5, 0.3) to (vmax);

                    \fill[thick, draw = black, gray!40] (vmin1) to (4, 1) to (3.8, 0.2) to (vmin1);

                    \fill[thick, draw = black, gray!40] (3, 2) to (5, 2) to (4.8, 0.5) to (3.2, 1);
                \end{scope}

                \draw[fill = white] plot [smooth cycle] coordinates { (0.5, 5) (1, 2.1) (5, 2) (5.5, 5)};
                \node[vertex] (u0) at (1.3, 2.8) {};
                \node[vertex] (u1) at (1.7, 4.3) {};
                \node[vertex] (u2) at (2.1, 3) {};
                \node[vertex] (u3) at (2.5, 2.4) {};
                \node[vertex] (u4) at (2.8, 3.9) {};
                \node[vertex] (u5) at (3.2, 4.6) {};
                \node[vertex] (u6) at (3.5, 2.5) {};
                \node[vertex] (u8) at (3.7, 3.7) {};
                \node[vertex] (u9) at (4.5, 4.2) {};
                \node[vertex] (u7) at (4.8, 3.7) {};
                \node at (0.7, 4.5) {$\X_1$};

                \begin{scope}[shift = {(2, 0)}]
                    \draw[decorate, decoration = {snake, amplitude = 1, segment length = 20}, fill = white] (2, 0.5) ellipse (1 and 0.6);
                    \node[vertex] (v0) at (1.5, 0.5) {};
                    \node[vertex] (v1) at (1.8, 0.9) {};
                    \node[vertex] (v2) at (2.1, 0.4) {};
                    \node at (2.5, 0.8) {$\I_1$};
                    \draw[densely dashed, red] (v0) to (v1) to (v2) to (v0);
                \end{scope}

                \draw (vmax) edge (vmin1);
            \end{tikzpicture}
            \caption{$j = 1; n_1 = 15; i_1 = 3; x_1 = 10$.}
            \label{fig:gapstr:restricted:decomp:k15:first}
        \end{subfigure}
            \hspace{0.01\textwidth}
        \begin{subfigure}[b]{0.45\textwidth}
            \centering
            \begin{tikzpicture}[vertex/.style = {draw, circle, fill = black, outer sep = 1pt, inner sep = 2pt}, scale = 0.8]
                \node[vertex, label = left:{\scriptsize $\vmax$}] (vmax) at (0, 1.5) {};
                \node[vertex, label = below:{\scriptsize $\vmin^1$}] (vmin1) at (1, 0) {};
                \node[vertex, label = below:{\scriptsize $\vmin^2$}] (vmin2) at (1.4, 2.75) {};

                \begin{scope}[on background layer]
                    \fill[thick, draw = black, gray!40] (vmax) to (0.3, 4.5) to (5, 2.5) to (vmax);
                    \fill[thick, draw = black, gray!40] (vmax) to (4, 1.15) to (3.5, 0.3) to (vmax);
                    \fill[thick, draw = black, gray!40] (vmin1) to (4, 1) to (3.8, 0.2) to (vmin1);
                    \fill[thick, draw = black, gray!40] (3, 2) to (5, 2) to (4.8, 0.5) to (3.2, 1);

                    \draw[fill = white] plot [smooth cycle] coordinates { (0.5, 5) (1, 2.1) (5, 2) (5.5, 5)};
                    \node at (0.7, 4.5) {$\X_1$};

                    \fill[thick, draw = black, gray!40] (vmin2) to (2.95, 2.75) to (2.5, 2.1) to (vmin2);
                    \fill[thick, draw = black, gray!40] (2.5, 2.5) to (3.5, 2.5) to (5, 3.4) to (2, 3.5);

                    \draw[fill = gray!10] plot [smooth cycle] coordinates { (1.5, 5) (2, 3.2) (5, 3.2) (5.3, 4.9)};
                    \node at (1.7, 4.6) {$\X_2$};
                \end{scope}

                \node[vertex] (u1) at (2.3, 4.3) {};
                \node[vertex] (u4) at (2.8, 3.5) {};
                \node[vertex] (u5) at (3.2, 4.6) {};
                \node[vertex] (u6) at (3.5, 2.5) {};
                \node[vertex] (u8) at (3.7, 3.7) {};
                \node[vertex] (u7) at (4, 4.5) {};
                \node[vertex] (u7) at (4.8, 3.7) {};

                \begin{scope}[shift = {(2, 0)}]
                    \draw[decorate, decoration = {snake, amplitude = 1, segment length = 20}, fill = white] (2, 0.5) ellipse (1 and 0.6);
                    \node[vertex] (v0) at (1.5, 0.5) {};
                    \node[vertex] (v1) at (1.8, 0.9) {};
                    \node[vertex] (v2) at (2.1, 0.4) {};
                    \node at (2.5, 0.8) {$\I_1$};
                    \draw[densely dashed, red] (v0) to (v1) to (v2) to (v0);
                \end{scope}

                \draw[decorate, decoration = {snake, amplitude = 0.8, segment length = 15}, fill = white] (3, 2.25) ellipse (0.8 and 0.5);
                \node[vertex] (w0) at (2.5, 2.3) {};
                \node[vertex] (w1) at (2.8, 2.65) {};
                \node[vertex] (w2) at (3, 2.2) {};
                \node at (3.4, 2.35) {$\I_2$};
                \draw[densely dashed, red] (w0) to (w1) to (w2) to (w0);

                \draw (vmax) edge (vmin1)
                            edge (vmin2);

            \end{tikzpicture}
            \caption{$j = 2; n_2 = 11; i_2 = 3; x_2 = 6$.}
            \label{fig:gapstr:restricted:decomp:k15:second}
        \end{subfigure}
            \par\bigskip
        \begin{subfigure}[b]{\textwidth}
            \centering
            \begin{tikzpicture}[vertex/.style = {draw, circle, fill = black, outer sep = 1pt, inner sep = 2pt}, scale = 0.8]
                \node[vertex, label = left:{\scriptsize $\vmax$}] (vmax) at (0, 1.5) {};
                \node[vertex, label = below:{\scriptsize $\vmin^1$}] (vmin1) at (1, 0) {};
                \node[vertex, label = below:{\scriptsize $\vmin^2$}] (vmin2) at (1.4, 2.75) {};
                \node[vertex, label = {[shift = {(0, -1)}]\scriptsize $\vmin^3$}] (vmin3) at (2.1, 4) {};
                \node[vertex, label = {[shift = {(0.45, -0.7)}]\scriptsize $\vmin^4$}] (vmin4) at (2.8, 4.2) {};

                \begin{scope}[on background layer]
                    \fill[thick, draw = black, gray!40] (vmax) to (0.3, 4.5) to (5, 2.5) to (vmax);
                    \fill[thick, draw = black, gray!40] (vmax) to (4, 1.15) to (3.5, 0.3) to (vmax);
                    \fill[thick, draw = black, gray!40] (vmin1) to (4, 1) to (3.8, 0.2) to (vmin1);
                    \fill[thick, draw = black, gray!40] (3, 2) to (5, 2) to (4.8, 0.5) to (3.2, 1);

                    \draw[fill = white] plot [smooth cycle] coordinates { (0.5, 5) (1, 2.1) (5, 2) (5.5, 5)};
                    \node at (0.7, 4.5) {$\X_1$};

                    \fill[thick, draw = black, gray!40] (vmin2) to (2.95, 2.75) to (2.5, 2.1) to (vmin2);
                    \fill[thick, draw = black, gray!40] (2.5, 2.5) to (3.5, 2.5) to (5, 3.4) to (2, 3.5);

                    \draw[fill = gray!10] plot [smooth cycle] coordinates { (1.5, 5) (2, 3.2) (5, 3.2) (5.3, 4.9)};
                    \node at (1.7, 4.6) {$\X_2$};

                    \fill[thick, draw = black, gray!40] (vmin3) to (3.3, 3.6) to (2.8, 3.3) to (vmin3);
                    \fill[thick, draw = black, gray!40] (3, 3.6) to (4, 3.4) to (5, 4) to (3, 4.5);

                    \draw[fill = gray!20] plot [smooth cycle] coordinates { (2.3, 4.95) (2.7, 4) (5, 3.9) (5.2, 4.75)};
                    \node at (2.5, 4.65) {$\X_3$};

                    \begin{scope}[shift = {(-0.2, 0.04)}]
                    \draw[fill = gray!30] plot [smooth cycle] coordinates { (3, 4.85) (3.1, 4.4) (3.7, 4.4) (3.7, 4.85)};
                    \node at (3.5, 4.625) {$\X_4$};
                    \node[vertex] (u7) at (3.2, 4.7) {};
                    \end{scope}
                \end{scope}

                \begin{scope}[shift = {(2, 0)}]
                    \draw[decorate, decoration = {snake, amplitude = 1, segment length = 20}, fill = white] (2, 0.5) ellipse (1 and 0.6);
                    \node[vertex] (v0) at (1.5, 0.5) {};
                    \node[vertex] (v1) at (1.8, 0.9) {};
                    \node[vertex] (v2) at (2.1, 0.4) {};
                    \node at (2.5, 0.8) {$\I_1$};
                    \draw[densely dashed, red] (v0) to (v1) to (v2) to (v0);
                \end{scope}

                \draw[decorate, decoration = {snake, amplitude = 0.8, segment length = 15}, fill = white] (3, 2.25) ellipse (0.8 and 0.5);
                \node[vertex] (w0) at (2.5, 2.3) {};
                \node[vertex] (w1) at (2.8, 2.65) {};
                \node[vertex] (w2) at (3, 2.2) {};
                \node at (3.4, 2.35) {$\I_2$};
                \draw[densely dashed, red] (w0) to (w1) to (w2) to (w0);

                \draw[decorate, decoration = {snake, amplitude = 0.8, segment length = 15}, fill = white] (3.3, 3.3) ellipse (0.9 and 0.3);
                \node[vertex] (x0) at (2.8, 3.4) {};
                \node[vertex] (x1) at (3.3, 3.5) {};
                \node at (3.8, 3.4) {$\I_3$};
                \draw[densely dashed, red] (x0) to (x1) to (x0);

                \draw[decorate, decoration = {snake, amplitude = 0.8, segment length = 15}, fill = white] (4.15, 4.1) ellipse (0.4 and 0.3);
                \node[vertex] (y0) at (4, 4.25) {};
                \node at (4.275, 4.25) {$\I_4$};

                \draw (vmax) edge (vmin1)
                            edge (vmin2)
                            edge[bend left = 5] (vmin3)
                            edge[out = 80, in = 160] (vmin4);

                    \draw (vmin4) edge[dashed, red] (u7) edge[bend left = 18] (y0);

            \end{tikzpicture}
            \caption{$j = 4; n_4 = 4; i_4 = 1; x_4 = 1$.}
            \label{fig:gapstr:restricted:decomp:k15:last}
        \end{subfigure}
        \caption{Decomposition process for $K_{15}$. Gray areas symbolize all edges connecting vertices in different sets. Observe that no $\vmin^j$ is adjacent to vertices in $V_{j+1}$.}
        \label{fig:gapstr:restricted:decomp:k15}
    \end{figure}
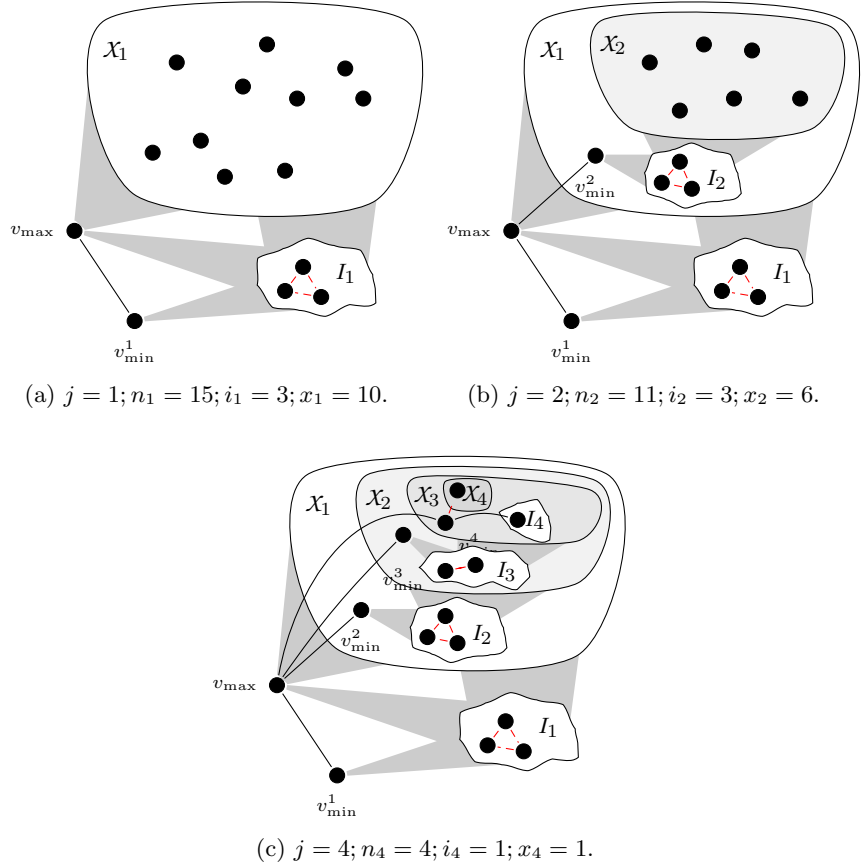

    We remark that since the recursive decomposition is done every time $x_j \ge 3$, the resulting configuration of the graph will have at most two vertices in set $\X_{j'}$ after the last iteration of the algorithm. 

    Now, we show that this resulting graph admits a gap-vertex-labelling by assigning labels to each vertex of $V(G)$ as follows:

    \begin{itemize}
        \item $\pi(\vmax) = 2^{n-1}$; 
        \item $\pi(\vmin^j) = 2^{j-1}$ for every $j \geq 1$; and
        \item $\pi(v) = 2^{n-2}$ for every $v \in \I_j$, $j \geq 1$.
    \end{itemize}

    It remains to assign labels to the vertices in $\X_{j'}$ of the last iteration $j'$. We remark that this set has either one or two vertices, by construction. Now, if $x_{j'} = 1$, assign label $2^{j'}$ to that vertex. Otherwise, there are exactly two vertices in $\X_{j'}$, and we assign labels $2^{j'}$ and $2^{j'+1}$ to these vertices, in any order. Colouring $c_\pi$ is defined as usual. In Figure~\ref{fig:restricted:gapstr:gvl:k15}, we exhibit a different representation of our restricted decomposition obtained from $K_{15}$. We also show our \gvl{} $(\pi, c_\pi)$. In the figure, vertices $v_j$ belong to set $\I_j$ and vertex $v_x$ is the singular vertex in $\X_4$. The removed edges are displayed as red, dashed lines between vertices. 

    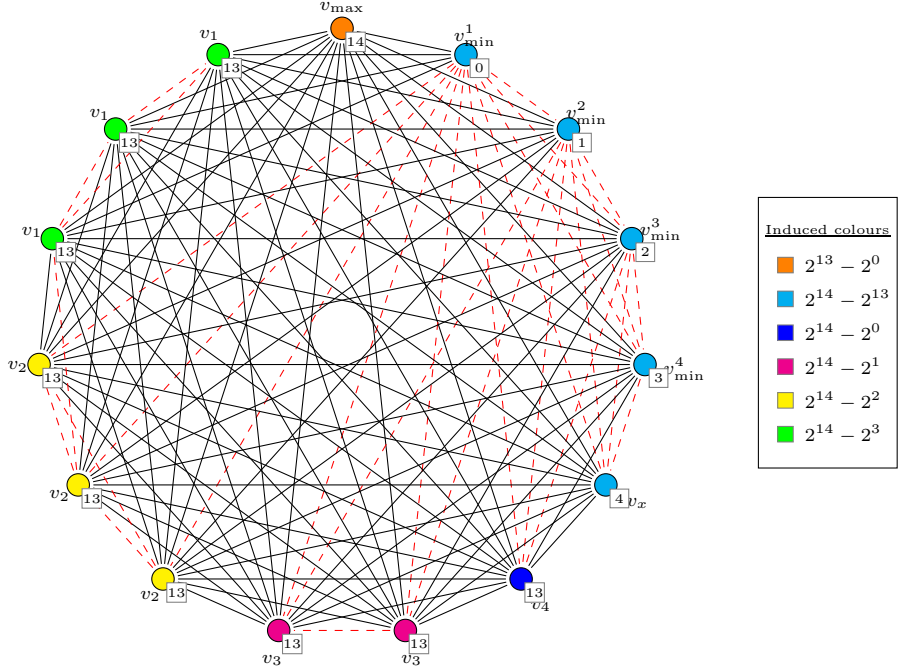
\begin{figure}[!htb]
        \centering
        \begin{adjustbox}{center}
        \begin{tikzpicture}[scale = 0.9]
                \def \radius {5};
                \begin{scope}[scale = 0.9]
                \node[vertex, fill = orange] (v0) at (90:\radius) {};
                \node[vertex, fill = green] (v14) at (114:\radius) {};
                \node[vertex, fill = green] (v13) at (138:\radius) {};
                \node[vertex, fill = green] (v12) at (162:\radius) {};
                \node[vertex, fill = yellow] (v11) at (186:\radius) {};
                \node[vertex, fill = yellow] (v10) at (210:\radius) {};
                \node[vertex, fill = yellow] (v9) at (234:\radius) {};
                \node[vertex, fill = magenta] (v8) at (258:\radius) {};
                \node[vertex, fill = magenta] (v7) at (282:\radius) {};
                \node[vertex, fill = blue] (v6) at (306:\radius) {};
                \node[vertex, fill = cyan] (v5) at (330:\radius) {};
                \node[vertex, fill = cyan] (v4) at (354:\radius) {};
                \node[vertex, fill = cyan] (v3) at (18:\radius) {};
                \node[vertex, fill = cyan] (v2) at (42:\radius) {};
                \node[vertex, fill = cyan] (v1) at (66:\radius) {};

                \node at (90:{\radius+0.35}) {\scriptsize $\vmax$};
                \node at (114:{\radius+0.35}) {\scriptsize $v_1$};
                \node at (138:{\radius+0.35}) {\scriptsize $v_1$};
                \node at (162:{\radius+0.35}) {\scriptsize $v_1$};
                \node at (186:{\radius+0.35}) {\scriptsize $v_{2}$};
                \node at (210:{\radius+0.35}) {\scriptsize $v_{2}$};
                \node at (234:{\radius+0.35}) {\scriptsize $v_{2}$};
                \node at (258:{\radius+0.5}) {\scriptsize $v_{3}$};
                \node at (282:{\radius+0.5}) {\scriptsize $v_{3}$};
                \node at (306:{\radius+0.55}) {\scriptsize $v_{4}$};
                \node at (330:{\radius+0.6}) {\scriptsize $v_{x}$};
                \node at (354:{\radius+0.65}) {\scriptsize $\vmin^{4}$};
                \node at (18:{\radius+0.5}) {\scriptsize $\vmin^{3}$};
                \node at (42:{\radius+0.4}) {\scriptsize $\vmin^{2}$};
                \node at (66:{\radius+0.35}) {\scriptsize $\vmin^{1}$};

                \draw (v0) edge (v1) edge (v2) edge (v3) edge (v4) edge (v5) edge (v6) edge (v7) edge (v8) edge (v9) edge (v10) edge (v11) edge (v12) edge (v13) edge (v14);

                \draw (v1) edge[red, dashed] (v2) edge[red, dashed] (v3) edge[red, dashed] (v4) edge[red, dashed] (v5) edge[red, dashed] (v6) edge[red, dashed] (v7) edge[red, dashed] (v8) edge[red, dashed] (v9) edge[red, dashed] (v10) edge[red, dashed] (v11) edge (v12) edge (v13) edge (v14);

                \draw (v2) edge[red, dashed] (v3) edge[red, dashed] (v4) edge[red, dashed] (v5) edge[red, dashed] (v6) edge[red, dashed] (v7) edge[red, dashed] (v8) edge (v9) edge (v10) edge (v11) edge (v12) edge (v13) edge (v14);

                \draw (v3) edge[red, dashed] (v4) edge[red, dashed] (v5) edge[red, dashed] (v6) edge (v7) edge (v8) edge (v9) edge (v10) edge (v11) edge (v12) edge (v13) edge (v14);

                \draw (v4) edge[red, dashed] (v5) edge (v6) edge (v7) edge (v8) edge (v9) edge (v10) edge (v11) edge (v12) edge (v13) edge (v14);

                \draw (v5) edge (v6) edge (v7) edge (v8) edge (v9) edge (v10) edge (v11) edge (v12) edge (v13) edge (v14);

                \draw (v6) edge (v7) edge (v8) edge (v9) edge (v10) edge (v11) edge (v12) edge (v13) edge (v14);

                \draw (v7) edge[red, dashed] (v8) edge (v9) edge (v10) edge (v11) edge (v12) edge (v13) edge (v14);

                \draw (v8) edge (v9) edge (v10) edge (v11) edge (v12) edge (v13) edge (v14);

                \draw (v9) edge[red, dashed] (v10) edge[red, dashed] (v11) edge (v12) edge (v13) edge (v14);

                \draw (v10) edge[red, dashed] (v11) edge[red, dashed] (v12) edge (v13) edge (v14);

                \draw (v11) edge (v12) edge (v13) edge (v14);

                \draw (v12) edge[red, dashed] (v13) edge[red, dashed] (v14);

                \draw (v13) edge[red, dashed] (v14);

                \begin{scope}[pi/.append style = {minimum width = 10, inner sep = -0.8, shift = {(0.18, -0.18)}}]
                    \node[pi] at (v0) {\tiny $14$};
                    \node[pi] at (v1) {\tiny $0$};
                    \node[pi] at (v2) {\tiny $1$};
                    \node[pi] at (v3) {\tiny $2$};
                    \node[pi] at (v4) {\tiny $3$};
                    \node[pi] at (v5) {\tiny $4$};
                    \node[pi] at (v6) {\tiny $13$};
                    \node[pi] at (v7) {\tiny $13$};
                    \node[pi] at (v8) {\tiny $13$};
                    \node[pi] at (v9) {\tiny $13$};
                    \node[pi] at (v10) {\tiny $13$};
                    \node[pi] at (v11) {\tiny $13$};
                    \node[pi] at (v12) {\tiny $13$};
                    \node[pi] at (v13) {\tiny $13$};
                    \node[pi] at (v14) {\tiny $13$};
                \end{scope}
                \end{scope}

                \begin{scope}[shift = {(7.25, -2)}]
                    \draw[fill = white] (-1.1, 0) rectangle (0.95, 4);
                    \node[align = center] at (-0.075, 3.5) {\tiny \underline{Induced colours}};

                    \node[pi, fill = orange] at (-0.7, 3) {\tiny \color{orange} $x$}; \node[anchor = west] at (-0.55, 3) {\scriptsize $2^{13} - 2^0$};

                    \node[pi, fill = cyan] at (-0.7, 2.5) {\tiny \color{cyan} $x$};\node[anchor = west] at (-0.55, 2.5) {\scriptsize $2^{14} - 2^{13}$};

                    \node[pi, fill = blue] at (-0.7, 2) {\tiny \color{blue} $x$};\node[anchor = west] at (-0.55, 0.5) {\scriptsize $2^{14} - 2^{3}$};

                    \node[pi, fill = magenta] at (-0.7, 1.5) {\tiny \color{magenta} $x$};\node[anchor = west] at (-0.55, 1) {\scriptsize $2^{14} - 2^{2}$};

                    \node[pi, fill = yellow] at (-0.7, 1) {\tiny \color{yellow} $x$};\node[anchor = west] at (-0.55, 1.5) {\scriptsize $2^{14} - 2^{1}$};

                    \node[pi, fill = green] at (-0.7, 0.5) {\tiny \color{green} $x$};\node[anchor = west] at (-0.55, 2) {\scriptsize $2^{14} - 2^{0}$};
                \end{scope}
            \end{tikzpicture}
            \end{adjustbox}
        \caption{Graph $G$ obtained by our decomposition of $K_{15}$, accompanied with the \gvl{} described in the text. The value $i$ in each box next to the vertices corresponds to the label $2^i$ assigned to that vertex. The induced colours are discriminated in the table to the right of the graph.}
        \label{fig:restricted:gapstr:gvl:k15}
    \end{figure}

    Let $f'(n)$ denote the number of edges removed in our construction. In order to complete the proof, we have to show: that colouring $c_\pi$ is a proper vertex-colouring of $G$; and that we removed $f'(n) \in \mathcal{O}(n\sqrt{n})$ edges from $K_n$. We start by showing the former. First, we draw the readers attention to the labels assigned to the vertices of $G$. The label set used in $\pi$ is\footnote{We remark that label $2^{j'+1}$ only belongs in this set if $x_{j'} = 2$ in part $\X_{j'}$ of the last iteration $j'$. Otherwise, the label set is $\{2^0, 2^1, \ldots, 2^{j'}, 2^{n-2}, 2^{n-1}\}$.} $\{2^0, 2^1, \ldots, 2^{j'}, 2^{j'+1}, 2^{n-2}$, $2^{n-1}\}$, where $j'$ denotes the last iteration of the recursive construction. Moreover, with the exception of $2^{n-2}$, i.e.\@ the label assigned to vertices $v \in \I_j$, $j \geq 1$, every label in the set is assigned to exactly one vertex. Now, consider $\vmax$ and observe that, since $\vmax$ is a universal vertex, the largest and smallest labels in $N(\vmax)$ are the largest and smallest label in $V(G)-\vmax$, namely $2^{n-2}$ and~$2^{0}$. We conclude that $c_\pi(\vmax) = 2 ^ {n-2} - 1$. 

    Next, we consider the vertices in each set $\I_j$, referring to these vertices as~$v_j$. Recall that, by construction, each $\I_j$ is an independent set. Also, every $v_j$ is adjacent to $\vmax$, which received label $2^{n-1}$. When $j = 1$, we have ${\vmax, \vmin^{1} \in N(v_1)}$, which induces ${c_\pi(v_1)= 2^{n-1} - 1}$. Hence, ${c_\pi(v_1) \neq c_\pi(\vmax)}$. For every $j \geq 2$, recall that vertices in~$\I_j$ are not adjacent to any $\vmin^{l}$, $l < j$, since $\I_j \in \X_{j-1}$. Moreover, $\pi(\vmin^j) < \pi(\vmin^{j+l})$ for all $j+l \leq j'$. Therefore, the smallest label in $N(v_j)$ is the label assigned to $\vmin^j$, and we conclude that $c_\pi(v_j) = 2^{n-1} - 2^{j-1}$ for every $v_j \in \I_j$. With the exception of $j = 1$, which we mention in the beginning of the paragraph, we conclude that $c_\pi(v_j)$ is always an even number. Therefore, $c_\pi(v_j) \neq c_\pi(\vmax)$ since $c_\pi(\vmax)$ is always odd.

    Now, consider the vertices in $\X_{j'}$. As previously stated, this set has either one or two vertices. First, suppose $|\X_{j'}| = 1$, and let $v_x$ be the vertex in this set. By construction,~$v_x$ is adjacent to: $\vmax$, which received label $2^{n-1}$; to every $v_j \in \I_j$, $1 \leq j \leq j'$, all of which received label $2^{n-2}$; and no other vertex. This implies that $c_\pi(v_x) = 2^{n-1} - 2^{n-2}$. Thus, $c_\pi(v_x) \neq c_\pi(w)$ for every $w \in N(v_x)$. Conversely, suppose $|\X_{j'}| = 2$, and let $v_x$ and $v_x'$ be the two vertices in $\X_{j'}$. Also, recall that $v_x$ and $v_x'$ received labels $2^{j}$ and $2^{j'+1}$, in any order. Without loss of generality, let $\pi(v_x) = 2^{j'}$. Now, since $\pi(v_x) < \pi(v_x') < \pi(w)$ for every other $w \in N(v_x)$ and $w \in N(v_x')$, it follows that $c_\pi(v_x) = 2^{n-1} - 2^{j'+1}$ and $c_\pi(v_x') = 2^{n-1} - 2^{j'}$. This, in turn, implies that $c_\pi(v_x) \neq c_\pi(v_x')$ and, moreover, that these induced colours do not conflict with that of the vertices in their respective neighbourhoods.

    Lastly, we consider the induced colours of vertices $\vmin^j$. For every $1 \leq j < j'$, we remark that $N(\vmin^j)$ consists only of $\vmax$ and vertices $v_j \in \I_j$; these vertices received labels $2^{n-1}$ and $2^{n-2}$, respectively. Then, we conclude that every $\vmin^j$ has colour $c_\pi(\vmin^j) = 2^{n-1} - 2^{n-2} = 2^{n-2}$. It follows that $c_\pi(\vmin^j) \neq c_\pi(\vmax)$. It is important to remark that the number of iterations $j' < n-1$ and, therefore, $c_\pi(\vmin) \neq c_\pi(v_j)$ for all $v_j \in \I_j$. We conclude that there are no conflicting vertices in $G$ and, consequently, that $c_\pi$ is a proper vertex-colouring of the graph. 

    Thus, it remains to prove that our construction removes $f'(n) \in \mathcal{O}(n\sqrt{n})$ from $K_n$. Equivalently, we show that $f'(n) \leq 3n\sqrt{n}$, for $n \in \N$. We prove this result by (strong) induction on $n$. When $n \leq 3$, the inequality naturally holds since $f'(n) = 0$. Now, suppose $f'(n') \leq 3 (n') \sqrt{n'}$ for every $1\le n' < n$, and let us consider the number $f'(n)$ of edges removed from $K_n$. Recalling that by our construction $i = \lfloor \sqrt{n} \rfloor$ and $x = n - i - 2$, we have
    \begin{align*}
        f'(n) &= x + \binom{i}{2} + f'(x+1) \nonumber\\
            &= \left(n - \lfloor \sqrt{n} \rfloor - 2\right) + \frac{\lfloor \sqrt{n} \rfloor(\lfloor \sqrt{n} \rfloor-1)}{2} + f'(n-\lfloor \sqrt{n} \rfloor -1).
    \end{align*}

    Clearly $n-\lfloor \sqrt{n} \rfloor -1 < n$ and we can apply the induction hypothesis:
    \begin{equation}
        \resizebox{0.9\hsize}{!}{
            $f'(n) \leq \left(n - \lfloor \sqrt{n} \rfloor - 2\right)+ \frac{\lfloor \sqrt{n} \rfloor(\lfloor \sqrt{n} \rfloor-1)}{2} + 3(n-\lfloor \sqrt{n} \rfloor -1)\sqrt{n-\lfloor \sqrt{n} \rfloor -1}$}.\label{eq:complicated}
    \end{equation}

    \noindent We use the following inequality to simplify equation~\eqref{eq:complicated}: $\sqrt{n} -1 \leq \lfloor \sqrt{n} \rfloor \leq \sqrt{n}$, considering three ``parts'' of the left-side equation. For the ``$x$ part'', we have
    \begin{align}
        n - \lfloor \sqrt{n} \rfloor - 2 &= n - (\lfloor \sqrt{n} \rfloor + 1) - 1\nonumber\\
        &\le n -\sqrt{n} -1\nonumber\\
        &< n - \sqrt{n} = \frac{2(n - \sqrt{n})}{2}.\label{eq:simple:x}
    \end{align}
    \noindent As for the ``$i$ part'', we have
    \begin{align}
        \frac{\lfloor\sqrt{n}\rfloor (\lfloor\sqrt{n}\rfloor - 1)}{2} \le \frac{n - \sqrt{n}}{2}.\label{eq:simple:i}
    \end{align}

    \noindent Finally, for the recursive part of the decomposition, we can apply the same logic as for the $x$ part and obtain
    \begin{align}      
        3(n - \lfloor\sqrt{n}\rfloor - 1)\sqrt{n - (\lfloor\sqrt{n}\rfloor + 1)} \le 3(n - \sqrt{n})\sqrt{n - \sqrt{n}}.\label{eq:simple:rec}
    \end{align}        

    Applying the inequalities from equations~\eqref{eq:simple:x},~\eqref{eq:simple:i} and~\eqref{eq:simple:rec} in~\eqref{eq:complicated}:
    \begin{align}
        f'(n) &\leq \frac{2(n - \sqrt{n})}{2}+ \frac{n - \sqrt{n}}{2} + 3(n - \sqrt{n})\sqrt{n - \sqrt{n}}\nonumber \\
        &\leq \frac{3}{2}(n - \sqrt{n})+ 3(n-\sqrt{n})\sqrt{n-\sqrt{n}}\nonumber\\
        &\leq \frac{3}{2}(n - \sqrt{n})+ \left(- 3\sqrt{n}\sqrt{n-\sqrt{n}} + 3n\sqrt{n-\sqrt{n}}\right). \label{eq:restricted:upper:eq2}
    \end{align}

    We draw the readers attention to the rightmost part of equation~\eqref{eq:restricted:upper:eq2}, remarking that $3n\sqrt{n-\sqrt{n}} \leq 3n\sqrt{n}$ for $n \geq 1$. Therefore, if the rest of the equation is at most zero, i.e. $\frac{3}{2}(n - \sqrt{n}) - 3\sqrt{n}\sqrt{n-\sqrt{n}} \leq 0$, then the desired result holds. For the sake of contradiction, suppose the contrary:
    \begin{align}
        \frac{3}{2}(n - \sqrt{n}) - 3\sqrt{n}\sqrt{n-\sqrt{n}} > 0 \nonumber &\iff \frac{3}{2}(n - \sqrt{n}) > 3\sqrt{n}\sqrt{n-\sqrt{n}} \nonumber \\
        &\iff n - \sqrt{n} > 2\sqrt{n}\sqrt{n-\sqrt{n}} \nonumber \\
        &\iff n^2 - 2n\sqrt{n} + n > 4n(n-\sqrt{n}) \nonumber \\
        &\iff 3n^2 - 2n\sqrt{n} - n < 0 \label{eq:restricted:upper:eq3}.
    \end{align}

    Since $n \geq 1$, we can divide equation~\eqref{eq:restricted:upper:eq3} by $n$, obtaining ${3n - 2\sqrt{n} - 1 \leq 0}$. This inequality is only satisfied when $0\leq n < 1$. However, since we are considering only $n \geq 1$, we conclude that
    \begin{align*}
        \frac{3}{2}(n - \sqrt{n}) - 3\sqrt{n}\sqrt{n-\sqrt{n}} &\leq 0\\
        \implies\frac{3}{2}(n - \sqrt{n}) - 1 - 3\sqrt{n}\sqrt{n-\sqrt{n}} + 3n\sqrt{n-\sqrt{n}} &\leq 3n\sqrt{n} \\
        \implies f'(n) &\leq 3n\sqrt{n}. 
    \end{align*}
    This completes the proof.
\end{proof}

\section{Concluding remarks and open problems}\label{sec:conclusion}

Graph labelling problems are remarkable. Each problem poses its own interesting questions and challenges. In particular, the gap-vertex-labelling problem is such a case: we have yet to determine the computational complexity of deciding whether a graph is gap-vertex-labelable or not. Our paper provides the first step to answering this question: we establish an upper bound of $\O(n^2)$ labels that can be used to properly label a given graph (if such a labelling exists). This result implies in an $\O(n!)$-time algorithm which decides whether a graph is gap-vertex-labelable or not.

On the other hand, we show three families of graphs which do not admit gap-vertex-labellings for any number of labels, namely the families of complete graphs and Powers of Paths and Cycles. Of course, much research can still be done in this area. We enlist a few open problems we consider interesting.

We have established that certain complete graphs $K_n$, powers of paths $P_{n}^{k}$ and powers of cycles $C_{n}^{k}$ do not admit any gap-vertex-labelling, regardless of the number of labels $k$ provided as input to the problem. However, for the ones that do admit such a labelling, we provide no bounds on the \underline{least} number of labels required, i.e., on their vertex-gap numbers.

\begin{problem*}
    Determine the vertex-gap numbers for the families of gap-vertex-labelable complete graphs $K_n$, powers of paths $P_n^k$ and powers of cycles $C_n^k$.
\end{problem*}

Since we have presented only results for classes of graphs, another open problem --- and perhaps the most interesting one of all --- would be to define what is the structural property a graph must have such that it is not gap-vertex-labelable. We believe that complete graph $K_4$ seems to be at the heart of this problem and that it may have a high correlation to the density of the graph.

\begin{problem*}
    Determine the structural property that determines whether a graph $G$ is gap-vertex-labelable.
\end{problem*}

Finally, we propose a novel parameter associated with this particular labelling problem. We investigate the gap-strength of complete graphs $K_n$ of order $n \ge 4$, establishing that one has to remove some number in between $\Omega(n^{1.2})$ and $\O(n^{1.5})$ of edges from $K_n$ in order to obtain a gap-vertex-labelable graph. The main open problem here would be to close this ``gap'' between the upper and lower bound. Furthermore, it would be interesting to analyse the gap-strength of other non-gap-vertex-labelable graphs, such as the families of powers of cycles and paths considered in this paper.

\subsection{Acknowledgements}

    We would like to thank Lucas Colucci for pointing us towards the Golomb and Sidom Rulers, without which we would not have established the upper-bound on the vertex-gap number of arbitrary graphs.

\bibliography{main-arxiv}

\newpage

\appendix


\section{Proofs and codes used in the proof of Proposition~\ref{prop:restricted}.}

    \subsection{Existence of global minimum of function $g(n, x)$}\label{app:restricted:claim1}

        We recall the definition of function $g(n, x)$ from Section~\ref{subsec:lb} below; we also show the first and second derivatives partial to $x$.
        \begin{align*}
            g(n, x) &= x + \frac{(n - x - 2)(n-x-3)}{2} + \frac{1}{10}(x+1)^{3/2}\\
            \frac{\partial g}{\partial x} &= x - n + \frac{7}{2} + \frac{3\sqrt{x+1}}{20}\\
            \frac{\partial^2 g}{\partial^2 x} &= 1 + \frac{3}{40\sqrt{x+1}}
        \end{align*}

        \begin{claim}\label{claim:lb:restricted:1}
            For all $n \ge 3$, there exists some $x^{*}$ such that for every $x$, $g(n, x) \ge g(n, x^{*})$.\hfill$\blacksquare$
        \end{claim}

        \begin{proof}
            By solving $\frac{\partial g}{\partial x} = 0$, we obtain min/max values of $x$ for $g(n, x)$.
            \begin{align*}
                x - n + \frac{7}{2} + \frac{3\sqrt{x+1}}{20} &= 0\\
                x &= \frac{1}{800}(-2791 + 800n \pm 3\sqrt{-3991 + 1600n})
            \end{align*}

            These values for $x$ were computed with the aid of the following Mathematica code.

            \begin{lstlisting}
                f[n_, x_] := x + (n - x - 2) (n - x - 3)/2 + (1/10)*(x + 1)^(3/2)
                g[n_, x_] := Simplify[D[f[n, x], x]]
                h[n_, x_] := Simplify[D[g[n, x], x]]
                Solve[g[n, x] == 0, x]
            \end{lstlisting}

            The output given by command line 4 above was

            \begin{lstlisting}
                {{x -> 1/800 (-2791 + 800 n - 3 Sqrt[-3991 + 1600 n])}, {x -> 
       1/800 (-2791 + 800 n + 3 Sqrt[-3991 + 1600 n])}}
            \end{lstlisting}

            Define $x_1$ and $x_2$ as follows: \[x_1 = \frac{1}{800}(-2791 + 800n - 3\sqrt{-3991 + 1600n}),\] \[x_2 = \frac{1}{800}(-2791 + 800n + 3\sqrt{-3991 + 1600n}).\] In fact, observe that both $x_1$ and $x_2$ are strictly increasing under values of $n$, which implies that the minimum values for $x_1$ and~$x_2$ occur also at the minimum value for $n$.

            In order to determine if they are local minima or maxima, we consider the second derivative $g''(n, x)$ in each of these points.
            \begin{align*}
                g''(5, x_1) &\approx 2.24345 \\
                g''(5, x_2) &\approx 2.36163
            \end{align*}

            The code used to obtain these results is presented below.

            \begin{lstlisting}
                x1 := 1/800 (-2791 + 800 n - 3 Sqrt[-3991 + 1600 n])
                x2 := 1/800 (-2791 + 800 n + 3 Sqrt[-3991 + 1600 n])
                N[f[5, x1]
                N[f[5, x2]
            \end{lstlisting}

            Note that both $g''(n, x_1), g''(n, x_2) \geq 0$ for any $n \geq \frac{3991}{1600} \approx 2.494375$. 

            \begin{lstlisting}
                Reduce[h[n, x1] >= 0, x]
                Reduce[h[n, x2] >= 0, x]
            \end{lstlisting}

            Since both values of the second derivative are positive, $x_1$ and $x_2$ are local minimum points for $g(n, x)$. Furthermore, since $g(n, x_2) > g(n, x_1)$, then $g(n, x_1)$ is a global minimum for $g(n, x)$. Thus, for any values of $n$ and ${x \ge 3}$, we have that $g(n, x) \geq g(n, x_1)$. This completes the proof of Claim~\ref{claim:lb:restricted:1}.\qedhere
        \end{proof}

    \subsection{Proving that the difference function is nonnegative and strictly increasing}\label{app:restricted:claim2}

        Let $g(n) = g(n, x^{*})$ be the function defined as follows, with $x^{*} = x_1$ from~\ref{app:restricted:claim1}.
        \begin{equation}
            \resizebox{\hsize}{!}{
            $g(n) = \frac{-2017919 + 647200n - 64000n^{3/2}-27\sqrt{-3991+1600n}+2\sqrt{2}\left(-1991+800n-3\sqrt{-3991+1600n}\right)^{3/2}}{640000}$}\nonumber
        \end{equation}

        Define $h(n)$ as the difference between $g(n)$ and $\frac{1}{10}n^{3/2}$ for increasing values of $n$, that is,
        \begin{align*}
            h(n) := g(n) - \frac{1}{10}n^{3/2}.
        \end{align*}

        \begin{claim}\label{claim:lb:restricted:2}
            For every $n \ge 4$, $h(n) \ge 0$ and $h'(n) \ge 0$.
        \end{claim}

        \begin{proof}

            If $\partial h/\partial n \ge 0$, this implies that the difference between these two functions is always increasing. That is, to show that $g(n) \in \Omega(n^{3/2})$ is equivalent to showing that $h(n) \ge 0$ and that $\partial h/\partial n \ge 0$, for all $n \ge 4$. In fact, for any $n \ge 2.49438$ this inequality holds, as was verified with the aid of Mathematica using the following code.\qedhere

            \begin{lstlisting}
                l[n_] := Simplify[f[n, x1]]
                gap[n_] := l[n] - (1/10)*n^(3/2)
                Reduce[gap[n] >= 0, n]
                Reduce[D[gap[n], n] >= 0, n]
             \end{lstlisting}

        \end{proof}


\section{Mathematica codes used in the proof of Theorem~\ref{theo:gapstrength:lowerbound}.}\label{app:lb:general}

    The following codes were used to validate equations~\eqref{eq:gapstri} and~\eqref{eq:gapstrx}.

    \begin{lstlisting}
        ri[i_] := i (i - 1)/2
        rx[x_] := x + ((x + 1)^(3/2))/10
        lowerq := (3/10)^(2/3)*q^(12/5)
        targetlb := (3/100)*q^(18/5)
        Reduce[ri[lowerq] >= targetlb, q]
        Reduce[rx[lowerq] - lowerq >= targetlb, q]
    \end{lstlisting}

    As for function $w(q)$ defined below (and in equation~\eqref{eq:lb:w}), we compute the limit when $q \rightarrow \infty$ and the first derivative in order to show that the lower-bound holds.

    \begin{lstlisting}
        w = Simplify[(2 z + d z^(6/5))/((3 q^(18/5))/100) /.z -> q^3 - 3*(3/10)^(2/3) q^(12/5), q >= 2]
        Limit[w, q -> Infinity]
        wprime = Simplify[D[w, {q}], q > 0]
        Reduce[3000 3^(2/3) 10^(1/3) + (-5000 + 27 3^(2/3) 10^(1/3)) q^(3/5) < 0  && q >= 6]
        
    \end{lstlisting}


\section{Results of Dynamic Programming.}\label{app:dp}

In Tables~\ref{table:dp:1} and~\ref{table:dp:2}, we present the results obtained by our Dynamic Programming Algorithm for $n \in \{4, \ldots, x\}$ and $n \in \{x+1, \ldots, 218\}$, respectively. Our algorithm computes a lower-bound on the least number $l'(n)$ of edges that must be removed in order to obtain a decomposition of $K_n$.

    \begin{table}[!h]
        \scriptsize
        \begin{adjustbox}{center}
    \begin{tabular}{*{4}{|c|c|c|}}
        \multicolumn{1}{c}{$n$} & \multicolumn{1}{c}{$l'$} & $\Omega$ & \multicolumn{1}{c}{$n$} & \multicolumn{1}{c}{$l'$} & $\Omega$ & \multicolumn{1}{c}{$n$} & \multicolumn{1}{c}{$l'$} & $\Omega$ & \multicolumn{1}{c}{$n$} & \multicolumn{1}{c}{$l'$} & \multicolumn{1}{c}{$\Omega$} \\\hline
        4  & 1  & 0.16 & 35 & 79  & 2.14 & 66 & 216 & 4.58 & 97 & 388 & 7.27 \\
5  & 2  & 0.21 & 36 & 83  & 2.21 & 67 & 221 & 4.66 & 98 & 394 & 7.36 \\
6  & 3  & 0.26 & 37 & 87  & 2.29 & 68 & 226 & 4.74 & 99 & 400 & 7.45 \\
7  & 4  & 0.31 & 38 & 91  & 2.36 & 69 & 231 & 4.83 & 100 & 406 & 7.54 \\
8  & 5  & 0.36 & 39 & 95  & 2.43 & 70 & 236 & 4.91 & 101 & 412 & 7.63 \\
9  & 7  & 0.42 & 40 & 99  & 2.51 & 71 & 241 & 5.00 & 102 & 418 & 7.72 \\
10 & 9  & 0.48 & 41 & 103 & 2.59 & 72 & 246 & 5.08 & 103 & 424 & 7.81 \\
11 & 11 & 0.53 & 42 & 107 & 2.66 & 73 & 251 & 5.17 & 104 & 430 & 7.90 \\
12 & 13 & 0.59 & 43 & 111 & 2.74 & 74 & 256 & 5.25 & 105 & 436 & 7.99 \\
13 & 15 & 0.65 & 44 & 115 & 2.81 & 75 & 261 & 5.34 & 106 & 442 & 8.08 \\
14 & 17 & 0.71 & 45 & 119 & 2.89 & 76 & 266 & 5.42 & 107 & 448 & 8.17 \\
15 & 19 & 0.77 & 46 & 123 & 2.97 & 77 & 271 & 5.51 & 108 & 454 & 8.26 \\
16 & 21 & 0.84 & 47 & 127 & 3.05 & 78 & 276 & 5.59 & 109 & 460 & 8.36 \\
17 & 23 & 0.90 & 48 & 131 & 3.12 & 79 & 281 & 5.68 & 110 & 466 & 8.45 \\
18 & 25 & 0.96 & 49 & 135 & 3.20 & 80 & 286 & 5.77 & 111 & 472 & 8.54 \\
19 & 28 & 1.03 & 50 & 139 & 3.28 & 81 & 292 & 5.85 & 112 & 478 & 8.63 \\
20 & 31 & 1.09 & 51 & 143 & 3.36 & 82 & 298 & 5.94 & 113 & 484 & 8.73 \\
21 & 34 & 1.16 & 52 & 147 & 3.44 & 83 & 304 & 6.03 & 114 & 490 & 8.82 \\
22 & 37 & 1.22 & 53 & 151 & 3.52 & 84 & 310 & 6.11 & 115 & 496 & 8.91 \\
23 & 40 & 1.29 & 54 & 156 & 3.60 & 85 & 316 & 6.20 & 116 & 502 & 9.00 \\
24 & 43 & 1.36 & 55 & 161 & 3.68 & 86 & 322 & 6.29 & 117 & 509 & 9.10 \\
25 & 46 & 1.43 & 56 & 166 & 3.76 & 87 & 328 & 6.38 & 118 & 516 & 9.19 \\
26 & 49 & 1.50 & 57 & 171 & 3.84 & 88 & 334 & 6.46 & 119 & 523 & 9.28 \\
27 & 52 & 1.57 & 58 & 176 & 3.92 & 89 & 340 & 6.55 & 120 & 530 & 9.38 \\
28 & 55 & 1.64 & 59 & 181 & 4.00 & 90 & 346 & 6.64 & 121 & 537 & 9.47 \\
29 & 58 & 1.71 & 60 & 186 & 4.08 & 91 & 352 & 6.73 & 122 & 544 & 9.57 \\
30 & 61 & 1.78 & 61 & 191 & 4.16 & 92 & 358 & 6.82 & 123 & 551 & 9.66 \\
31 & 64 & 1.85 & 62 & 196 & 4.25 & 93 & 364 & 6.91 & 124 & 558 & 9.76 \\
32 & 67 & 1.92 & 63 & 201 & 4.33 & 94 & 370 & 7.00 & 125 & 565 & 9.85 \\
33 & 71 & 1.99 & 64 & 206 & 4.41 & 95 & 376 & 7.09 & 126 & 572 & 9.94 \\
34 & 75 & 2.06 & 65 & 211 & 4.49 & 96 & 382 & 7.18 & 127 & 579 & 10.04 \\
        \hline
            \end{tabular}
    \end{adjustbox}
    \caption{Results from our Dynamic Programming algorithm for ${n \in \{4, \ldots, 127\}}$.}
    \label{table:dp:1}
\end{table}

\begin{table}[!h]
        \scriptsize
        \begin{adjustbox}{center}
    \begin{tabular}{*{3}{|c|c|c|}}
        \multicolumn{1}{c}{$n$} & \multicolumn{1}{c}{$l'$} & $\Omega$ & \multicolumn{1}{c}{$n$} & \multicolumn{1}{c}{$l'$} & $\Omega$ & \multicolumn{1}{c}{$n$} & \multicolumn{1}{c}{$l'$} & \multicolumn{1}{c}{$\Omega$} \\\hline
        128 & 586 & 10.13 & 159 & 803 & 13.15 & 190 & 1050 & 16.28 \\
129 & 593 & 10.23 & 160 & 810 & 13.25 & 191 & 1058 & 16.38 \\
130 & 600 & 10.32 & 161 & 818 & 13.34 & 192 & 1066 & 16.48 \\
131 & 607 & 10.42 & 162 & 826 & 13.44 & 193 & 1074 & 16.59 \\
132 & 614 & 10.52 & 163 & 834 & 13.54 & 194 & 1082 & 16.69 \\
133 & 621 & 10.61 & 164 & 842 & 13.64 & 195 & 1090 & 16.79 \\
134 & 628 & 10.71 & 165 & 850 & 13.74 & 196 & 1098 & 16.90 \\
135 & 635 & 10.80 & 166 & 858 & 13.84 & 197 & 1106 & 17.00 \\
136 & 642 & 10.90 & 167 & 866 & 13.94 & 198 & 1114 & 17.10 \\
137 & 649 & 10.99 & 168 & 874 & 14.04 & 199 & 1122 & 17.21 \\
138 & 656 & 11.09 & 169 & 882 & 14.14 & 200 & 1130 & 17.31 \\
139 & 663 & 11.19 & 170 & 890 & 14.24 & 201 & 1138 & 17.42 \\
140 & 670 & 11.28 & 171 & 898 & 14.35 & 202 & 1146 & 17.52 \\
141 & 677 & 11.38 & 172 & 906 & 14.45 & 203 & 1154 & 17.62 \\
142 & 684 & 11.48 & 173 & 914 & 14.55 & 204 & 1162 & 17.73 \\
143 & 691 & 11.58 & 174 & 922 & 14.65 & 205 & 1170 & 17.83 \\
144 & 698 & 11.67 & 175 & 930 & 14.75 & 206 & 1178 & 17.94 \\
145 & 705 & 11.77 & 176 & 938 & 14.85 & 207 & 1186 & 18.04 \\
146 & 712 & 11.87 & 177 & 946 & 14.95 & 208 & 1194 & 18.15 \\
147 & 719 & 11.96 & 178 & 954 & 15.05 & 209 & 1202 & 18.25 \\
148 & 726 & 12.06 & 179 & 962 & 15.15 & 210 & 1210 & 18.36 \\
149 & 733 & 12.16 & 180 & 970 & 15.26 & 211 & 1218 & 18.46 \\
150 & 740 & 12.26 & 181 & 978 & 15.36 & 212 & 1226 & 18.57 \\
151 & 747 & 12.36 & 182 & 986 & 15.46 & 213 & 1234 & 18.67 \\
152 & 754 & 12.45 & 183 & 994 & 15.56 & 214 & 1242 & 18.78 \\
153 & 761 & 12.55 & 184 & 1002 & 15.66 & 215 & 1250 & 18.88 \\
154 & 768 & 12.65 & 185 & 1010 & 15.77 & 216 & 1259 & 18.99 \\
155 & 775 & 12.75 & 186 & 1018 & 15.87 & 217 & 1268 & 19.09 \\
156 & 782 & 12.85 & 187 & 1026 & 15.97 & 218 & 1277 & 19.20 \\\cline{7-9}
157 & 789 & 12.95 & 188 & 1034 & 16.07 & \multicolumn{3}{c}{} \\
158 & 796 & 13.05 & 189 & 1042 & 16.18 & \multicolumn{3}{c}{} \\

        \cline{1-6}
            \end{tabular}
    \end{adjustbox}
    \caption{Results from our Dynamic Programming algorithm for ${n \in \{128, \ldots, 218\}}$.}
    \label{table:dp:2}
\end{table}

\paragraph{Explanation} The tables consist of (main) columns, separated by the doubled lines. In each of these columns, we present three values. First, the order of a complete graph $K_n$; second, the lower-bound $l'(n)$ on the least number of edges one needs to remove from $K_n$ in order to obtain a decomposition $G(\X, \Y, \Z, \I)$ of $K_n$, as computed by our program; and third, the value of $\frac{3}{100}n^{1.2}$, which corresponds to the lower-bound of $\gapstr(K_n)$ as demonstrated by Theorem~\ref{theo:gapstrength:lowerbound}. (We have shortened the second and third values in the table headers as $l'$ and $\Omega$, respectively, for simplicity.) Observe that for all $n \in \{4, \ldots, 218\}$, $l' \ge \Omega$.

\end{document}